\documentclass[a4paper,envcountsame]{llncs}
\usepackage{url}
\spnewtheorem{dfn}{Definition}{\bfseries}{\itshape}
\spnewtheorem{lem}{Lemma}{\bfseries}{\itshape}
\spnewtheorem{thm}{Theorem}{\bfseries}{\itshape}
\spnewtheorem{cor}{Corollary}{\bfseries}{\itshape}
\spnewtheorem{prp}{Proposition}{\bfseries}{\itshape}
\spnewtheorem*{rmk}{Remark}{\itshape}{}
\spnewtheorem{invariant}[lemma]{Invariant}{\bfseries}{\itshape}

\usepackage[utf8]{inputenc}
\usepackage{microtype}
\usepackage{amssymb,amsmath}
\usepackage[basic]{complexity}
\usepackage{thmtools}
\usepackage{times}
\usepackage{float}

\usepackage{xcolor}
\usepackage{xspace}
\usepackage{paralist}
\usepackage{comment}

\usepackage{booktabs}
\usepackage{todonotes}

\usepackage{graphicx}
\usepackage{subfig}

\usepackage{tikz}
\usetikzlibrary{arrows}
\usetikzlibrary{calc}
\usetikzlibrary{positioning}
\usetikzlibrary{shapes}
\usetikzlibrary{automata}
\usetikzlibrary{decorations.pathmorphing}

\usepackage[ruled,vlined,linesnumbered,nokwfunc,algo2e]{algorithm2e_}
\DontPrintSemicolon
\SetAlCapFnt{\normalfont\scshape}
\SetProcNameSty{textsc}
\SetFuncSty{textsc}

\newcommand{\set}[1]{\{#1\}}
\newcommand{\lu}{\textup{(}}
\newcommand{\ru}{\textup{)}}
\newcommand{\upbr}[1]{\lu #1\ru}

\newcommand{\sseq}{\langle v_0,v_1,v_2,\ldots\rangle}
\newcommand{\pat}{\omega\xspace} 
\newcommand{\Pat}{\Omega\xspace}
\newcommand{\obj}{\phi\xspace}
\newcommand{\Inf}{\mathrm{Inf}\xspace}
\newcommand{\objsty}[2]{\textrm{#1}(#2)}

\newcommand{\Reach}[1]{\objsty{Reach}{#1}}

\newcommand{\Streett}[1]{\objsty{Streett}{#1}}

\newcommand{\SP}{\mathrm{TP}\xspace}
\newcommand{\idxs}{j\xspace} 

\newcommand{\win}[1]{\llangle 1 \rrangle\left(#1\right)}
\newcommand{\scc}{C\xspace}
\newcommand{\inscc}{\expandafter\MakeLowercase\expandafter{\scc}\xspace} 
\DeclareMathOperator{\Out}{\textit{Out}\xspace}
\DeclareMathOperator{\In}{\textit{In}\xspace}

\newcommand{\sccalg}{\ProcNameSty{allSCCs}}

\newcommand{\reachG}[2]{\ProcNameSty{GraphReach}(#1, #2)} 
\newcommand{\goodC}{\DataSty{goodC}} 

\newcommand{\at}[3]{\textit{Attr}_{#1}(#2, #3)\xspace}

\newcommand{\straa}{\sigma\xspace}
\newcommand{\Straa}{\Sigma\xspace}
\newcommand{\pr}[3]{\mathrm{Pr}^{#1}_{#2}\left(#3\right)}
\makeatletter
\newsavebox{\@brx}
\newcommand{\llangle}[1][]{\savebox{\@brx}{\(\m@th{#1\langle}\)}%
  \mathopen{\copy\@brx\mkern2mu\kern-0.9\wd\@brx\usebox{\@brx}}}
\newcommand{\rrangle}[1][]{\savebox{\@brx}{\(\m@th{#1\rangle}\)}%
  \mathclose{\copy\@brx\mkern2mu\kern-0.9\wd\@brx\usebox{\@brx}}}
\makeatother
\newcommand{\as}[1]{\llangle 1 \rrangle_\textit{as}\left(#1\right)}
\newcommand{\mdp}{P\xspace}
\newcommand{\vp}{V_1\xspace}
\newcommand{\vr}{V_R\xspace}
\newcommand{\trans}{\delta\xspace}
\newcommand{\target}{T\xspace}
\newcommand{\intarget}{\expandafter\MakeLowercase\expandafter{\target}\xspace}
\newcommand{\ec}{X\xspace}
\newcommand{\inec}{\expandafter\MakeLowercase\expandafter{\ec}\xspace}
\newcommand{\mecalg}{\ProcNameSty{allMECs}}
\newcommand{\mectime}{\textsc{mec}\xspace}
\newcommand{\good}{\DataSty{goodEC}}

\newcommand{\rout}{\textit{rout}}

\newcommand{\dombound}[2]{h\xspace} 

\newcommand{\badv}{B\xspace}

\newcommand{\bad}{\ProcNameSty{Bad}}

\newcommand{\pre}{\mathsf{Pre}\xspace}
\newcommand{\post}{\mathsf{Post}\xspace}
\newcommand{\cpre}[1]{\mathsf{CPre}_{#1}\xspace}

\begin{document}
\mainmatter

\title{Symbolic Algorithms for Graphs and Markov Decision Processes with Fairness Objectives\vspace{-2mm}}
\author{Krishnendu Chatterjee$^1$, Monika Henzinger$^2$, Veronika Loitzenbauer$^3$, \\ Simin Oraee$^4$, and Viktor Toman$^1$\vspace{-1mm}}
\institute{$^1$ IST Austria, Klosterneuburg, Austria \\
			$^2$ University of Vienna, Austria \\
            	$^3$ Johannes Kepler University Linz, Austria \\
            	$^4$ Max Planck Institute for Software Systems, Kaiserslautern, Germany}

\maketitle
\vspace{-4mm}
\begin{abstract}
Given a model and a specification, the fundamental model-checking problem asks
for algorithmic verification of whether the model satisfies the specification.
We consider graphs and Markov decision processes (MDPs), which are fundamental 
models for reactive systems.
One of the very basic specifications that arise in verification of reactive systems
is the strong fairness (aka Streett) objective. Given different types of 
requests and corresponding grants, the objective requires that for each type, if the 
request event happens infinitely often, then the corresponding grant event 
must also happen infinitely often. 
All $\omega$-regular objectives can be expressed as Streett objectives
and hence they are canonical in verification. 
To handle the state-space explosion, symbolic algorithms are required that operate on a
succinct implicit representation of the system rather than explicitly accessing
the system.
While explicit algorithms for graphs and MDPs with Streett objectives have been 
widely studied, there has been no improvement of the basic symbolic algorithms.
The worst-case numbers of symbolic steps required for the basic symbolic algorithms 
are as follows: quadratic for graphs and cubic for MDPs.
In this work we present the first sub-quadratic symbolic algorithm for graphs 
with Streett objectives, and our algorithm is sub-quadratic even for MDPs.
Based on our algorithmic insights we present an implementation of the new symbolic
approach and show that it improves the existing approach on several academic benchmark
examples.
\end{abstract}

\vspace{-4mm}
\section{Introduction}\label{sec:intro}
\vspace{-1mm}

In this work we present faster symbolic algorithms for graphs and Markov decision 
processes (MDPs) with strong fairness objectives.
For the fundamental model-checking problem, the input consists of a {\em model} and 
{\em a specification}, and the algorithmic verification problem is to check whether 
the model {\em satisfies} the specification. 
We first describe the specific model-checking problem we consider and then our 
contributions. 

\vspace{-0.5mm}
\smallskip\noindent{\em Models: Graphs and MDPs.}
Two standard models for reactive systems are graphs and Markov decision processes (MDPs).
Vertices of a graph represent states of a reactive system, edges represent transitions 
of the system, and infinite paths of the graph represent non-terminating trajectories 
of the reactive system.
MDPs extend graphs with probabilistic transitions that represent reactive systems with 
uncertainty. 
Thus graphs and MDPs are the de-facto model of reactive systems with nondeterminism,
and nondeterminism with stochastic aspects, respectively~\cite{ClarkeBook,baierbook}.

\vspace{-0.5mm}
\smallskip\noindent{\em Specification: Strong Fairness (aka Streett) Objectives.}
A basic and fundamental property in the analysis of reactive systems is the 
{\em strong fairness condition}, which informally requires that if events are enabled infinitely often, then 
they must be executed infinitely often.
More precisely, the strong fairness conditions (aka Streett objectives) consist of 
$k$ types of requests and corresponding grants, and the objective requires that for each 
type if the request happens infinitely often, then the corresponding grant must also happen
infinitely often. 
After safety, reachability, and liveness, the strong fairness condition is one of the most
standard properties that arise in the analysis of reactive systems, and chapters 
of standard textbooks in verification are devoted to it (e.g., 
\cite[Chapter~3.3]{ClarkeBook},~\cite[Chapter~3]{MPProgress},~\cite[Chapters~8,~10]{AlurHenzingerBook}). 
Moreover, all $\omega$-regular objectives can be described by Streett objectives, e.g., 
LTL formulas and non-deterministic $\omega$-automata can be translated to
deterministic Streett automata~\cite{Safra88} and efficient translation
has been an active research area~\cite{ChatterjeeGK13,EsparzaK14,KomarkovaK14}. 
Thus Streett objectives are a canonical class of objectives that arise in verification.

\vspace{-0.5mm}
\smallskip\noindent{\em Satisfaction.} 
The basic notions of satisfaction for graphs and MDPs are as follows:
For graphs the notion of satisfaction requires that there is a trajectory (infinite path) 
that belongs to the set of paths described by the Streett objective.
For MDPs the satisfaction requires that there is a policy to resolve the nondeterminism 
such that the Streett objective is ensured almost-surely (with probability~1).
Thus the algorithmic model-checking problem of graphs and MDPs with Streett objectives is a 
core problem in verification.

\vspace{-0.5mm}
\smallskip\noindent{\em Explicit vs Symbolic Algorithms.}
The traditional algorithmic studies consider {\em explicit} algorithms that 
operate on the explicit representation of the system. In contrast, 
{\em implicit} or {\em symbolic} algorithms only use a set of predefined operations 
and do not explicitly access the system~\cite{ClarkeGP99}.
The significance of symbolic algorithms in verification is as follows:
to combat the state-space explosion, large systems must be succinctly represented 
implicitly and then symbolic algorithms are scalable, 
whereas explicit algorithms do not scale as it is computationally too expensive
to even explicitly construct the system.

\vspace{-0.5mm}
\smallskip\noindent{\em Relevance.} 
In this work we study symbolic algorithms for graphs and MDPs with Streett objectives.
Symbolic algorithms for the analysis of graphs and MDPs are at the heart of many state-of-the-art
tools such as SPIN, NuSMV for graphs~\cite{SPIN,NUSMV} and PRISM, LiQuor, Storm for MDPs~\cite{PRISM,LIQUOR,STORM}.
Our contributions are related to the algorithmic complexity of graphs and MDPs with 
Streett objectives for symbolic algorithms. 
We first present previous results and then our contributions.

\vspace{-0.5mm}
\smallskip\noindent{\em Previous Results.}
The most basic algorithm for the problem for graphs is based on repeated SCC (strongly
connected component) computation, and informally can be described as follows:
for a given SCC, (a)~if for every request type that is present in the SCC 
the corresponding grant type is also present in the SCC, then the SCC is identified 
as ``good'', (b)~else vertices of each request type that has no corresponding
grant type in the SCC are removed, and the algorithm recursively proceeds
on the remaining graph.
Finally, reachability to good SCCs is computed. 
The current best-known symbolic algorithm for SCC computation requires $O(n)$ symbolic
steps, for graphs with $n$ vertices~\cite{GentiliniPP08}, and moreover, the algorithm is optimal~\cite{ChatterjeeDHL18}.
For MDPs, the SCC computation has to be replaced by MEC (maximal end-component) computation,
and the current best-known symbolic algorithm for MEC computation requires $O(n^2)$ symbolic steps.
While there have been several explicit algorithms for graphs with Streett 
objectives~\cite{HenzingerT96,ChatterjeeHL15}, MEC computation~\cite{ChatterjeeH11,ChatterjeeH12,ChatterjeeH14}, 
and MDPs with Streett objectives~\cite{ChatterjeeDHL16}, as well as symbolic algorithms for MDPs with 
B\"uchi objectives~\cite{ChatterjeeHJS13}, the current best-known bounds for symbolic algorithms 
with Streett objectives are obtained from the basic algorithms, which are $O(n \cdot \min(n,k))$ 
for graphs and $O(n^2\cdot \min(n,k))$ for MDPs, where $k$ is the number of types of request-grant pairs.

\vspace{-0.5mm}
\smallskip\noindent{\em Our Contributions.}
In this work our main contributions are as follows:
\begin{compactitem}
\item We present a symbolic algorithm that requires $O(n \cdot \sqrt{m \log n})$ symbolic steps,
both for graphs and MDPs, where $m$ is the number of edges. 
In the case $k=O(n)$, the previous worst-case bounds are quadratic ($O(n^2)$) for graphs and 
cubic ($O(n^3)$) for MDPs.
In contrast, we present the first sub-quadratic symbolic algorithm both for graphs as well as MDPs.
Moreover, in practice, since most graphs are sparse (with $m=O(n)$), the worst-case bounds of 
our symbolic algorithm in these cases are $O(n \cdot \sqrt{n\log n})$.
Another interesting contribution of our work is that we also present an $O(n \cdot \sqrt{m})$ symbolic 
steps algorithm for MEC decomposition, which is relevant for our results as well as of 
independent interest, as MEC decomposition is used in many other algorithmic problems 
related to MDPs.
Our results are summarized in Table~\ref{tab:symb:comparison}.

\item While our main contribution is theoretical, based on the algorithmic insights we 
also  present a new symbolic algorithm implementation for graphs and MDPs with Streett objectives. 
We show that the new algorithm improves (by around 30\%) the basic algorithm 
on several academic benchmark examples from the VLTS benchmark suite~\cite{VLTS}.
\end{compactitem}

\vspace{-1mm}
\begin{table}
\renewcommand{\arraystretch}{1.3}
\caption{Symbolic algorithms for Streett objectives and MEC decomposition.}\label{tab:symb:comparison}
\vspace{0.5mm}
\centering
\begin{tabular}{@{}llll@{}}
\toprule
& \multicolumn{2}{c}{Symbolic Operations}\\
\cmidrule{2-4}
Problem & Basic Algorithm & Improved Algorithm & Reference \\
\midrule
Graphs with Streett  & $O(n \cdot \min(n, k))$ & \boldmath$\mathbf{O(n \sqrt{m \log n})}$ & Theorem~\ref{thm:improvedgraphs}\\
MDPs with Streett  & $O(n^2 \cdot \min(n, k))$ & \boldmath$\mathbf{O(n \sqrt{m \log n})}$ & Theorem~\ref{thm:improvedmdps} \\
MEC decomposition & $O(n^2)$ & \boldmath$\mathbf{O(n \sqrt{m})}$ & Theorem~\ref{thm:improvedmecs} \\
\bottomrule
\end{tabular}
\end{table}
\setlength\intextsep{0em}
\smallskip

\vspace{-0.5mm}
\smallskip\noindent{\em Technical Contributions.}
The two key technical contributions of our work are as follows:
\begin{compactitem}
\item \emph{Symbolic Lock Step Search:} We search for newly emerged SCCs by a 
local graph exploration around vertices that lost adjacent edges. 
In order to find small new SCCs first, all searches are conducted ``in parallel'', i.e., in lock-step, 
and the searches stop as soon as the first one finishes successfully. 
This approach has successfully been used to improve explicit algorithms~\cite{HenzingerT96,ChatterjeeJH03,ChatterjeeH12,ChatterjeeDHL16}.
Our contribution is a non-trivial symbolic variant (Section~\ref{sec:lss}) 
which lies at the core of the theoretical improvements.

\item \emph{Symbolic Interleaved MEC Computation:}
For MDPs the identification of vertices that have to be removed can be interleaved with the computation of MECs such that in each iteration the computation of SCCs instead of MECs is sufficient to 
make progress~\cite{ChatterjeeDHL16}. We present a symbolic variant of this interleaved computation. 
This interleaved MEC computation is the basis for applying the lock-step search to MDPs.
\end{compactitem}

\vspace{-3mm}
\section{Definitions}\label{sec:defin}

\vspace{-2mm}
\subsection{Basic Problem Definitions}

\vspace{-1mm}
\smallskip\noindent{\em Markov Decision Processes \upbr{MDPs} and Graphs.}
An MDP~$\mdp = ((V, E), \allowbreak(\vp, \vr), \allowbreak\trans)$ consists of a
finite directed graph $G = (V, E)$ with a set of $n$ vertices~$V$ and a set of $m$ 
edges~$E$, a partition of the vertices into 
\emph{player~1 vertices} $\vp$ and \emph{random vertices} $\vr$, and a 
probabilistic transition function~$\trans$. We call an edge $(u,v)$ with
$u \in \vp$ a \emph{player~1 edge} and an edge $(v, w)$ with $v \in \vr$ a
\emph{random edge}. For $v \in V$ we define $\In(v)=\{w \in V \mid (w,v) \in E\}$
and $\Out(v)=\{w \in V \mid (v,w) \in E\}$.
The probabilistic transition function is a function from $\vr$ to $\mathcal{D}(V)$, 
where $\mathcal{D}(V)$ is the set of probability distributions over $V$ and 
a random edge $(v, w) \in E$ if and only if $\trans(v)[w] > 0$.
Graphs are a special case of MDPs with $\vr = \emptyset$.

\vspace{-0.5mm}
\smallskip\noindent{\em Plays and Strategies.} 
A \emph{play} or infinite path in $\mdp$ is an infinite sequence $\pat = \langle v_0, 
v_1, v_2, \ldots \rangle$ such that $(v_i, v_{i+1}) \in E$ for all $i \in \mathbb{N}$;
we denote by $\Pat$ the set of all plays.
A player~1 \emph{strategy}~$\straa: V^* \cdot \vp \rightarrow V$ is a function that 
assigns to every finite prefix~$\pat \in V^* \cdot \vp$ of a play that ends in a 
player~1 vertex~$v$ a successor vertex $\straa(\pat) \in V$ such that 
$(v, \straa(\pat)) \in E$; we denote by $\Straa$ the 
set of all player~1 strategies. A strategy is \emph{memoryless} if we have 
$\straa(\pat) = \straa(\pat')$ for any $\pat, \pat' \in V^* \cdot \vp$ that 
end in the same vertex $v \in \vp$.

\vspace{-0.5mm}
\smallskip\noindent{\em Objectives.}
An \emph{objective} $\obj$ is a subset of $\Pat$ said to be winning
for player~1. We say that a  play $\pat \in \Pat$
\emph{satisfies the objective} if $\pat \in \obj$. For a vertex set~$\target
\subseteq V$ the \emph{reachability objective} is the set of infinite paths
that contain a vertex of $\target$, i.e., 
$\Reach{\target} = \set{\sseq \in \Pat \mid \exists j \ge 0: v_j \in \target}$.
Let $\Inf(\pat)$ for $\pat \in \Pat$ denote the set of vertices that occur
infinitely often in $\pat$. Given a set $\SP$ of $k$ pairs $(L_i, U_i)$ of vertex
sets $L_i, U_i \subseteq V$ with $1 \le i \le k$, the \emph{Streett objective}
is the set of infinite paths for which it holds
\emph{for each} $1 \le i \le k$ that whenever a vertex of $L_i$ occurs
infinitely often, then a vertex of $U_i$ occurs infinitely often, i.e.,
$\Streett{\SP} = \set{\pat \in \Pat \mid L_i \cap \Inf(\pat) = \emptyset
\text{ or } U_i \cap \Inf(\pat) \ne \emptyset \text{ for all } 1 \le i \le k}$.

\vspace{-0.5mm}
\smallskip\noindent{\em Almost-Sure Winning Sets.}
For any measurable set of plays $A \subseteq \Pat$ we denote by
$\pr{\straa}{v}{A}$ the probability that a play starting at $v \in V$
belongs to $A$ when player~1 plays strategy~$\straa$. 
A strategy~$\straa$ is \emph{almost-sure \upbr{a.s.} winning} from a vertex
$v \in V$ for an objective $\obj$ if $\pr{\straa}{v}{\obj} = 1$.
The \emph{almost-sure winning set} $\as{\mdp, \obj}$
of player~1 is the set of vertices for which player~1 has an
almost-sure winning strategy. In graphs the existence of an almost-sure
winning strategy corresponds to the existence of a play in the objective,
and the set of vertices for which player~1 has an (almost-sure) winning
strategy is called the \emph{winning set} $\win{\mdp, \obj}$ of player~1.

\vspace{-0.5mm}
\smallskip\noindent{\em Symbolic Encoding of MDPs.}
Symbolic algorithms operate on sets of vertices, which are usually described by 
Binary Decision Diagrams (\textsc{bdd}s)~\cite{Lee59,Akers78}.
In particular Ordered Binary Decision Diagrams~\cite{Bryant85} (\textsc{Obdd}s) 
provide a canonical symbolic representation of Boolean functions. 
For the computation of almost-sure winning sets of MDPs it is sufficient to encode MDPs 
with \textsc{Obdd}s and one additional bit that denotes whether a vertex is 
in~$\vp$ or~$\vr$.

\vspace{-0.5mm}
\smallskip\noindent{\em Symbolic Steps.}
One symbolic step corresponds to one primitive operation as supported by 
standard symbolic packages like \textsc{CuDD}~\cite{Somenzi15}.
In this paper we only allow the same basic \emph{set-based symbolic operations} as 
in~\cite{RaviBS00,GentiliniPP03,BloemGS06,ChatterjeeHJS13}, namely set operations and
the following one-step symbolic operations for a set of vertices $Z$:
(a)~the one-step predecessor operator
$\pre(Z)=\{v \in V \mid \Out(v) \cap Z \ne \emptyset\};$
(b)~the one-step  successor operator
\mbox{$\post(Z)=\{v \in V \mid \In(v) \cap Z \ne \emptyset\};$}
and (c)~the one-step \emph{controllable} predecessor operator
$\cpre{R}(Z) = \left\{ v \in \vp \mid \Out(v) \subseteq Z \right\}
\cup  \left\{ v \in \vr \mid \Out(v) \cap Z \ne \emptyset \right\};$
i.e., the $\cpre{R}$ operator computes all vertices such that the successor belongs to~$Z$ with
positive probability. This operator can be defined using the $\pre$ operator
and basic set operations as follows:
$\cpre{R}(Z)= \pre(Z) \setminus (\vp \cap \pre(V \setminus Z))\,$.
We additionally allow cardinality computation and picking an 
arbitrary vertex from a set as in~\cite{ChatterjeeHJS13}.

\vspace{-0.5mm}
\smallskip\noindent{\em Symbolic Model.} Informally, a symbolic algorithm does not 
operate on explicit representation of the transition function of a graph, 
but instead accesses it through $\pre$ and $\post$ operations.
For explicit algorithms, a $\pre/\post$ operation on a set of vertices (resp., a single vertex)
requires $O(m)$ (resp., the order of indegree/outdegree of the vertex) time. 
In contrast, for symbolic algorithms $\pre/\post$ operations are considered unit-cost.
Thus an interesting algorithmic question is whether better algorithmic bounds can be obtained
considering $\pre/\post$ as unit operations.
Moreover, the basic set operations are computationally less expensive 
(as they encode the relationship between the state variables) compared to the $\pre/\post$ 
symbolic operations (as they encode the transitions and thus the relationship between the 
present and the next-state variables). 
In all presented algorithms, the number of set operations is asymptotically at most the 
number of $\pre/\post$ operations. 
Hence in the sequel we focus on the number of $\pre/\post$ operations of algorithms.

\vspace{-0.5mm}
\smallskip\noindent{\em Algorithmic Problem.}
Given an MDP~$\mdp$ (resp. a graph~$G$) and a set of Streett pairs $\SP$, the problem
we consider asks for a symbolic algorithm to compute the almost-sure winning
set $\as{\mdp, \Streett{\SP}}$ (resp. the winning set $\win{G, \Streett{\SP}}$), which
is also called the \emph{qualitative analysis} of MDPs (resp. graphs).

\vspace{-2mm}
\subsection{Basic Concepts related to Algorithmic Solution}

\vspace{-1mm}
\smallskip\noindent{\em Reachability.}
For a graph~$G = (V, E)$ and a set of vertices $S \subseteq V$ 
the set $\reachG{G}{S}$ is the set of vertices of $V$ that \emph{can 
reach} a vertex of $S$ within~$G$, and it can be identified with at most
$\lvert \reachG{G}{S} \setminus S \rvert + 1$ many $\pre$ operations.

\vspace{-0.5mm}
\smallskip\noindent{\em Strongly Connected Components.}
For a set of vertices $S \subseteq V$ we denote by 
$G[S] = (S, E \cap (S \times S))$ the subgraph of the graph~$G$ induced by 
the vertices of~$S$. An induced subgraph~$G[S]$ is strongly connected 
if there exists a path in~$G[S]$ between every pair of vertices of $S$.
A \emph{strongly connected component \upbr{SCC}} of~$G$ is a set of 
vertices~$\scc\subseteq V$ such that the induced subgraph~$G[\scc]$ is 
strongly connected and $\scc$ is a maximal set in~$V$ with this property.
We call an SCC \emph{trivial} if it only contains a single 
vertex and no edges; and \emph{non-trivial} otherwise. The SCCs of~$G$ partition
its vertices and can be found in $O(n)$ symbolic steps~\cite{GentiliniPP08}.
A bottom SCC~$\scc$ in a directed graph~$G$ is an SCC with no edges from 
vertices of~$\scc$ to vertices of~$V \setminus \scc$, i.e., an SCC without
\emph{outgoing} edges. Analogously, a top SCC~$\scc$ is an SCC with no \emph{incoming}
edges from~$V \setminus \scc$.
For more intuition for bottom and top SCCs, consider the graph in which 
each SCC is contracted into a single vertex (ignoring edges within
an SCC). In the resulting directed acyclic graph the sinks represent the
bottom SCCs and the sources represent the top SCCs.
Note that every graph has at least one bottom and at least one top SCC.
If the graph is not strongly connected, then there exist at least one
top and at least one bottom SCC that are disjoint and thus one of them 
contains at most half of the vertices of~$G$.

\vspace{-0.5mm}
\smallskip\noindent{\em Random Attractors.}
In an MDP~$\mdp$ the \emph{random attractor} $\at{R}{\mdp}{W}$  of a set of vertices 
$W$ is defined as $\at{R}{\mdp}{W} =
\bigcup_{j \ge 0} Z_j$ where $Z_0 = W$ and $Z_{j+1} = Z_j \cup \cpre{R}(Z_j)$
for all $j > 0$. The attractor can be computed with at most
$\lvert \at{R}{\mdp}{W} \setminus  W \rvert + 1$ many $\cpre{R}$ operations.

\vspace{-0.5mm}
\smallskip\noindent{\em Maximal End-Components.}
Let $\ec$ be a vertex set without outgoing random edges, i.e.,
with $\Out(v) \subseteq \ec$ for all $v \in \ec \cap \vr$.
A sub-MDP of an MDP~$\mdp$ induced by a vertex set $\ec \subseteq V$
without outgoing random edges is defined as 
$\mdp[\ec] = ((\ec, E \cap (\ec \times \ec), (\vp \cap \ec, 
\vr \cap \ec), \trans)$. Note that the requirement that $\ec$ has
no outgoing random edges
is necessary in order to use the same probabilistic transition function~$\trans$.
An \emph{end-component} \upbr{EC} of an MDP~$\mdp$ is a set of 
vertices $\ec \subseteq V$ such that \upbr{a} $\ec$ has no outgoing random 
edges, i.e., $\mdp[\ec]$ is a valid sub-MDP, \upbr{b} the induced sub-MDP $\mdp[\ec]$
is strongly connected, and \upbr{c} $\mdp[\ec]$ contains at least one edge.
Intuitively, an end-component is a set of vertices for which player~1 can ensure
that the play stays within the set and almost-surely reaches all the vertices in the 
set (infinitely often). An end-component is a \emph{maximal end-component} \upbr{MEC} 
if it is maximal under set inclusion.
An end-component is \emph{trivial} if it consists of a single vertex \upbr{with
a self-loop}, otherwise it is \emph{non-trivial}.
The \emph{MEC decomposition} of an MDP consists of all MECs of the MDP.

\vspace{-0.5mm}
\smallskip\noindent{\em Good End-Components.}
All algorithms for MDPs with Streett objectives are based on finding 
good end-components, defined below. Given the union of all good end-components, the 
almost-sure winning set is obtained by computing the 
almost-sure winning set for the reachability objective
with the union of all good end-components as the target set.
The correctness of this approach is shown in \cite{ChatterjeeDHL16,Loitzenbauer16}
(see also~\cite[Chap.~10.6.3]{baierbook}). For Streett objectives a good end-component
is defined as follows. In the special case of graphs they are called good components.

\begin{dfn}[Good end-component]
	Given an MDP $\mdp$ and a set $\SP= \{(L_\idxs, U_\idxs) \mid 1 \le \idxs \le k\}$ of
	target pairs, a \emph{good end-component} is an 
	end-component $\ec$ of $\mdp$ such that for each $1 \le \idxs \le k$ either 
	$L_\idxs \cap \ec = \emptyset$ or $U_\idxs \cap \ec \ne \emptyset$.
	A maximal good end-component is a good end-component that is maximal with respect to set inclusion.
\end{dfn}

\begin{lem}[{Correctness of Computing Good End-Components~\cite[Corollary~2.6.5, Proposition~2.6.9]{Loitzenbauer16}}]
\label{lem:gec}
	For an MDP~$\mdp$ and a set~$\SP$ of target pairs, let $\mathcal{\ec}$ be the set of all maximal good end-components.
	Then $\as{\mdp, \Reach{\bigcup_{\ec \in \mathcal{\ec}}\ec}}$ is equal to $\as{\mdp, \Streett{\SP}}$.
\end{lem}

\vspace{-0.5mm}
\smallskip\noindent{\em Iterative Vertex Removal.}
All the algorithms for Streett objectives maintain vertex sets that are 
candidates for good end-components. For such a vertex set~$S$ we (a) 
refine the maintained sets according to the SCC decomposition of $\mdp[S]$
and (b) for a set of vertices~$W$ for which we know that it cannot be contained 
in a good end-component, we remove its random attractor from $S$. The following lemma 
shows the correctness of these operations.

\vspace{-0.5mm}
\begin{lem}[{Correctness of Vertex Removal~\cite[Lemma~2.6.10]{Loitzenbauer16}}]
\label{lem:eccontained}
	Given an MDP $\mdp = ((V, E), (\vp, \vr), \trans)$, let $\ec$ be an end-component with $X \subseteq S$ for 
	some $S \subseteq V$. Then 
	\begin{compactitem}
	 \item[\upbr{a}]$\ec \subseteq \scc$ for one SCC~$\scc$ of $\mdp[S]$ and
	 \item[\upbr{b}] $\ec \subseteq S \setminus \at{R}{\mdp'}{W}$ for each
     $W \subseteq V \setminus \ec$ and each sub-MDP~$\mdp'$ containing~$\ec$.
	\end{compactitem}
\end{lem}

\vspace{-0.5mm}
Let $\ec$ be a good end-component. Then $\ec$ is an end-component and for each index $\idxs$,
$\ec \cap U_\idxs = \emptyset$ implies $\ec \cap L_\idxs = \emptyset$ . 
Hence we obtain the following corollary.

\vspace{-0.5mm}
\begin{cor}[{\cite[Corollary~4.2.2]{Loitzenbauer16}}]\label{cor:geccontained}
Given an MDP $\mdp$, let $\ec$ be a \emph{good} end-component with 
$X \subseteq S$ for some $S \subseteq V$. 
For each $i$ with $S \cap U_i = \emptyset$ it holds that
$\ec \subseteq S \setminus \at{R}{\mdp[S]}{ L_i \cap S}$.
\end{cor}

\vspace{-0.5mm}
For an index~$\idxs$ with $S \cap U_\idxs = \emptyset$ we call the 
vertices of $S \cap L_\idxs$ \emph{bad vertices}. The set of all bad 
vertices $\bad(S) = \bigcup_{1 \le i \le k} \set{v \in L_i \cap S \mid 
U_i \cap S = \emptyset}$ can be computed with $2 k$ set operations.

\vspace{-3mm}
\section{Symbolic Divide-and-Conquer with Lock-Step Search}\label{sec:lss}

\vspace{-1mm}
In this section we present a symbolic version of the lock-step
search for strongly connected subgraphs~\cite{HenzingerT96}. 
This symbolic version is used in all subsequent results, i.e., the
sub-quadratic symbolic algorithms for graphs and MDPs with Streett objectives,
and for MEC decomposition.

\vspace{-1mm}
\smallskip\noindent{\em Divide-and-Conquer.}
The common property of the algorithmic problems we consider in this work is
that the goal is to identify subgraphs of the input graph $G = (V, E)$
that are strongly connected and satisfy some additional properties. The difference between
the problems lies in the required additional properties. We 
describe and analyze the Procedure~\ref{proc:lockstep} that we use in all our improved 
algorithms to efficiently implement a divide-and-conquer approach based on the
requirement of strong connectivity, that is, we divide a subgraph~$G[S]$,
induced by a set of vertices~$S$, into 
two parts that are not strongly connected within $G[S]$ or detect that $G[S]$ is 
strongly connected. 

\vspace{-1mm}
\smallskip\noindent{\em Start Vertices of Searches.}
The input to Procedure~\ref{proc:lockstep} is a set of vertices~$S \subseteq V$ 
and two subsets of $S$ denoted by~$H_S$ and~$T_S$.
In the algorithms that call the procedure as a subroutine, vertices
contained in~$H_S$ have lost incoming edges (i.e., they were a  ``head'' of
a lost edge) and vertices contained in~$T_S$ have lost outgoing edges
(i.e., they were a  ``tail'' of a lost edge) since the last time a superset
of $S$ was identified as being strongly connected.
For each vertex~$h$ of $H_S$ the procedure
conducts a backward search (i.e., a sequence of $\pre$ operations) within~$G[S]$ to 
find the vertices of $S$ that can 
reach~$h$; and analogously a forward search (i.e., a sequence of $\post$ operations)
from each vertex~$t$ of $T_S$ is conducted.

\vspace{-1mm}
\smallskip\noindent{\em Intuition for the Choice of Start Vertices.}
If the subgraph~$G[S]$ is not strongly connected, then it contains at least 
one top SCC and at least one bottom SCC that are disjoint. Further, if for 
a superset $S' \supset S$ the subgraph~$G[S']$ was strongly connected, then 
each top SCC of $G[S]$ contains a vertex that had an additional 
incoming edge in $G[S']$ compared to $G[S]$, and analogously each bottom SCC of 
$G[S]$ contains a vertex that had an additional outgoing edge. Thus by keeping
track of the vertices that lost incoming or outgoing edges, the following invariant
will be maintained by all our improved algorithms.

\vspace{-1mm}
\begin{invariant}[Start Vertices Sufficient]\label{inv:HT} We have $H_S, T_S \subseteq S$.
	Either \upbr{a} $H_S \cup T_S = \emptyset$ and 
	$G[S]$ is strongly connected or \upbr{b} at least one vertex of
	each top SCC of $G[S]$ is contained in $H_S$ and
	at least one vertex of 
	each bottom SCC of $G[S]$ is contained in $T_S$.
\end{invariant}

\begin{procedure}
\caption{Lock-Step-Search($G$, $S$, $H_S$, $T_S$)}
\label{proc:lockstep}
		\tcc{$\pre$ and $\post$ defined w.r.t.\ to $G$}
		\lForEach{$v \in H_S \cup T_S$}{
			$C_v \gets \set{v}$\
		}
		\While{true}{
			$H'_S \gets H_S$, $T'_S \gets T_S$\;
			\ForEach(\tcc*[h]{search for top SCC}){$h \in H_S$}{
					$\scc'_h \gets (\scc_h \cup \pre(\scc_h)) \cap S$\;
					\lIf{$\lvert \scc'_h \cap H'_S \rvert > 1$}{
						$H'_S \gets H'_S \setminus \set{h}$
					}\Else{
						\lIf{$\scc'_h = \scc_h$}{
						\Return ($\scc_h$, $H'_S$, $T_S$)
						}
						$\scc_h \gets \scc'_h$\;
					}
			}
			\ForEach(\tcc*[h]{search for bottom SCC}){$t \in T_S$}{
				$\scc'_t \gets (\scc_t \cup \post(\scc_t)) \cap S$\;
				\lIf{$\lvert \scc'_t \cap T'_S \rvert > 1$}{
					$T'_S \gets T'_S \setminus \set{t}$
				}\Else{
					\lIf{$\scc'_t = \scc_t$}{
						\Return ($\scc_t$, $H'_S$, $T'_S$)
					}
					$\scc_t \gets \scc'_t$\;
				}
			}
			$H_S \gets H'_S$, $T_S \gets T'_S$\;
		}
\end{procedure}

\vspace{-1mm}
\smallskip\noindent{\em Lock-Step Search.}
The searches from the vertices of $H_S \cup T_S$
are performed in \emph{lock-step}, that is, (a) one step is performed
in each of the searches before the next step of any search is done and (b) all 
searches stop as soon as the first of the searches finishes. 
This is implemented in Procedure~\ref{proc:lockstep} as follows.
A step in the search from a vertex $t \in T_S$ 
(and analogously for $h \in H_S$)
corresponds to the execution of the iteration of the 
for-each loop for $t \in T_S$. In an iteration of a for-each loop 
we might discover that we do not need to consider this search further 
(see the paragraph on ensuring strong connectivity below)
and update the set $T_S$ (via $T'_S$) for future iterations accordingly.
Otherwise the set $C_t$ is either strictly increasing in this step of the 
search or the search for $t$ terminates and we return the set of vertices 
in $G[S]$ that are reachable from $t$.
So the two for-each loops over the vertices of $T_S$ and $H_S$ 
that are executed in an iteration of the while-loop perform
one step of each of the searches and the while-loop stops as 
soon as a search stops, i.e., a return statement is executed and hence 
this implements properties~(a) and~(b) of lock-step search.
Note that the while-loop terminates, i.e., a return statement is executed
eventually because for all $t \in T_S$ 
(and resp.\ for all $h \in H_S$) the sets $C_t$ are 
monotonically increasing over the iterations of the while-loop, 
we have $C_t \subseteq S$, and if some set $C_t$ does not increase 
in an iteration, then it is either removed from $T_S$ and thus not 
considered further or a return statement is executed. 
Note that when a search from a vertex $t \in T_S$ stops, it has discovered a 
maximal set of vertices~$\scc$ that can be reached from $t$; and analogously
for $h \in H_S$. Figure~\ref{fig:lssex} shows a small intuitive
example of a call to the procedure.

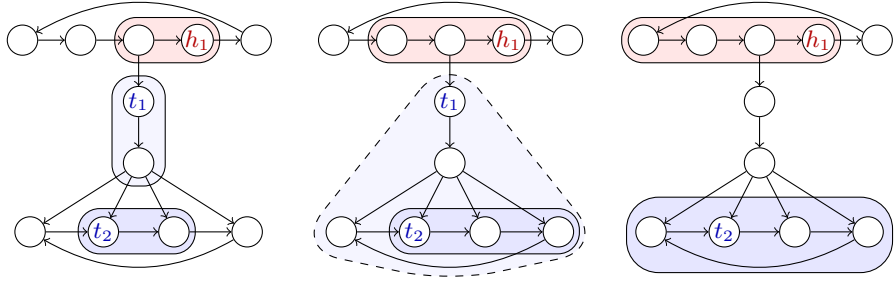
\begin{figure}
\begin{footnotesize}
\begin{center}
\begin{tikzpicture}[circle,minimum size=0.4cm,inner sep=0mm,fill=white,->]

	\node (0-BGh1) [draw,rectangle,rounded corners=8pt,fill=red!10,
    minimum height=0.6cm, minimum width=1.38cm] {};
	\node (0-a1) [draw,left=1cm of 0-BGh1] {};
    \node (0-a2) [draw,right=0.35cm of 0-a1] {};
    \node (0-a3) [draw,right=0.35cm of 0-a2,fill=white] {};
    \node (0-a4) [draw,right=0.35cm of 0-a3,fill=white,text=black!30!red] {$h_1$};
    \node (0-a5) [draw,right=0.35cm of 0-a4] {};
	\node (0-BGt1) [draw,rectangle,rounded corners=8pt,fill=blue!04,
    minimum height=1.46cm, minimum width=0.7cm, below=0.26cm of 0-a3] {};
    \node (0-BGt2) [draw,rectangle,rounded corners=8pt,fill=blue!10,
    minimum height=0.6cm, minimum width=1.54cm,
    below right=2.08cm and 0.58cm of 0-a1] {};
    \node (0-b1) [draw,below=0.4cm of 0-a3,text=black!30!blue,fill=white] {$t_1$};
    \node (0-c1) [draw,below=0.4cm of 0-b1,fill=white] {};
    \node (0-dd) [below=0.5cm of 0-c1] {};
    \node (0-d2) [draw,left=0.05cm of 0-dd,text=black!30!blue,fill=white] {$t_2$};
    \node (0-d1) [draw,left=0.55cm of 0-d2,fill=white] {};
    \node (0-d3) [draw,right=0.05cm of 0-dd,fill=white] {};
    \node (0-d4) [draw,right=0.55cm of 0-d3,fill=white] {};
    
    \draw [] (0-a1) to[] (0-a2);
    \draw [] (0-a2) to[] (0-a3);
    \draw [] (0-a3) to[] (0-a4);
    \draw [] (0-a4) to[] (0-a5);
    \draw [] (0-a5) to[bend right] (0-a1);
    \draw [] (0-a3) to[] (0-b1);
    \draw [] (0-b1) to[] (0-c1);
    \draw [] (0-c1) to[] (0-d1);
    \draw [] (0-c1) to[] (0-d2);
    \draw [] (0-c1) to[] (0-d3);
    \draw [] (0-c1) to[] (0-d4);
    \draw [] (0-d1) to[] (0-d2);
    \draw [] (0-d2) to[] (0-d3);
    \draw [] (0-d3) to[] (0-d4);
    \draw [] (0-d4) to[bend left] (0-d1);
    
	\node (1-BGh1) [draw,rectangle,rounded corners=8pt,fill=red!10,
    minimum height=0.6cm, minimum width=2.14cm,right=1.26cm of 0-a5] {};
	\node (1-a1) [draw,left=0.24cm of 1-BGh1] {};
    \node (1-a2) [draw,right=0.35cm of 1-a1,fill=white] {};
    \node (1-a3) [draw,right=0.35cm of 1-a2,fill=white] {};
    \node (1-a4) [draw,right=0.35cm of 1-a3,fill=white,text=black!30!red] {$h_1$};
    \node (1-a5) [draw,right=0.35cm of 1-a4] {};
    \draw[dashed,fill=blue!04,rounded corners=18pt]
        (1-BGh1.south)+(0,0.14) -- +(-2.06,-2.46) -- +(0,-2.9)
        -- +(2.06,-2.46) -- cycle;
    \node (1-BGt2) [draw,rectangle,rounded corners=8pt,fill=blue!10,
    minimum height=0.6cm, minimum width=2.5cm,
    right=2.54cm of 0-BGt2] {};
    \node (1-b1) [draw,below=0.4cm of 1-a3,text=black!30!blue,fill=white] {$t_1$};
    \node (1-c1) [draw,below=0.4cm of 1-b1,fill=white] {};
    \node (1-dd) [below=0.5cm of 1-c1] {};
    \node (1-d2) [draw,left=0.05cm of 1-dd,text=black!30!blue,fill=white] {$t_2$};
    \node (1-d1) [draw,left=0.55cm of 1-d2,fill=white] {};
    \node (1-d3) [draw,right=0.05cm of 1-dd,fill=white] {};
    \node (1-d4) [draw,right=0.55cm of 1-d3,fill=white] {};
    
    \draw [] (1-a1) to[] (1-a2);
    \draw [] (1-a2) to[] (1-a3);
    \draw [] (1-a3) to[] (1-a4);
    \draw [] (1-a4) to[] (1-a5);
    \draw [] (1-a5) to[bend right] (1-a1);
    \draw [] (1-a3) to[] (1-b1);
    \draw [] (1-b1) to[] (1-c1);
    \draw [] (1-c1) to[] (1-d1);
    \draw [] (1-c1) to[] (1-d2);
    \draw [] (1-c1) to[] (1-d3);
    \draw [] (1-c1) to[] (1-d4);
    \draw [] (1-d1) to[] (1-d2);
    \draw [] (1-d2) to[] (1-d3);
    \draw [] (1-d3) to[] (1-d4);
    \draw [] (1-d4) to[bend left] (1-d1);
    
	\node (2-BGh1) [draw,rectangle,rounded corners=8pt,fill=red!10,
    minimum height=0.6cm, minimum width=2.88cm,right=0.5cm of 1-a5] {};
	\node (2-a1) [draw,left=-0.5cm of 2-BGh1,fill=white] {};
    \node (2-a2) [draw,right=0.35cm of 2-a1,fill=white] {};
    \node (2-a3) [draw,right=0.35cm of 2-a2,fill=white] {};
    \node (2-a4) [draw,right=0.35cm of 2-a3,fill=white,text=black!30!red] {$h_1$};
    \node (2-a5) [draw,right=0.35cm of 2-a4] {};
    \node (2-b1) [draw,below=0.4cm of 2-a3] {};
    \node (2-c1) [draw,below=0.4cm of 2-b1] {};
    \node (2-BGt2) [draw,rectangle,rounded corners=12pt,fill=blue!10,
    minimum height=1.0cm, minimum width=3.5cm,
    below=0.24cm of 2-c1] {};
    \node (2-dd) [below=0.5cm of 2-c1] {};
    \node (2-d2) [draw,left=0.05cm of 2-dd,text=black!30!blue,fill=white] {$t_2$};
    \node (2-d1) [draw,left=0.55cm of 2-d2,fill=white] {};
    \node (2-d3) [draw,right=0.05cm of 2-dd,fill=white] {};
    \node (2-d4) [draw,right=0.55cm of 2-d3,fill=white] {};
    
    \draw [] (2-a1) to[] (2-a2);
    \draw [] (2-a2) to[] (2-a3);
    \draw [] (2-a3) to[] (2-a4);
    \draw [] (2-a4) to[] (2-a5);
    \draw [] (2-a5) to[bend right] (2-a1);
    \draw [] (2-a3) to[] (2-b1);
    \draw [] (2-b1) to[] (2-c1);
    \draw [] (2-c1) to[] (2-d1);
    \draw [] (2-c1) to[] (2-d2);
    \draw [] (2-c1) to[] (2-d3);
    \draw [] (2-c1) to[] (2-d4);
    \draw [] (2-d1) to[] (2-d2);
    \draw [] (2-d2) to[] (2-d3);
    \draw [] (2-d3) to[] (2-d4);
    \draw [] (2-d4) to[bend left] (2-d1);    
    
\end{tikzpicture}
\end{center}
\end{footnotesize}
\vspace{-6mm}
\caption{An example of symbolic lock-step search showing the first three
iterations of the main while-loop. Note that during the second iteration,
the search started from $t_1$ is disregarded since it collides with $t_2$.
In the subsequent fourth iteration, the search started from $t_2$ is
returned by the procedure.}
\label{fig:lssex}
\end{figure}

\vspace{-1mm}
\smallskip\noindent{\em Comparison to Explicit Algorithm.}
In the \emph{explicit} version of the algorithm~\cite{HenzingerT96,ChatterjeeDHL16} 
the search from vertex~$t \in T_S$ performs a depth-first search
that terminates exactly when every \emph{edge} reachable from $t$ is explored.
Since any search that starts outside of a bottom SCC but reaches the bottom SCC 
has to explore more edges than the search started inside of the bottom SCC, 
the first search from a vertex of $T_S$ that terminates has exactly explored 
(one of) the smallest (in the number of edges) bottom SCC(s) of $G[S]$. Thus on
explicit graphs the explicit lock-step search from the vertices of $H_S \cup T_S$
finds (one of) the smallest (in the number of edges) top or bottom 
SCC(s) of $G[S]$ in time proportional to the number of searches times the 
number of edges in the identified SCC. 
In \emph{symbolically} represented graphs it can happen (1) that
a search started outside of a bottom (resp.\ top) SCC terminates earlier than the 
search started within the bottom (resp.\ top) SCC and (2) that a search
started in a larger (in the number of vertices) top or bottom SCC terminates
before one in a smaller top or bottom SCC. We discuss next how we address these
two challenges.

\vspace{-1mm}
\smallskip\noindent{\em Ensuring Strong Connectivity.}
First, we would like the set returned by Procedure~\ref{proc:lockstep} 
to indeed be a top or bottom SCC of $G[S]$. For this we use the following observation
for bottom SCCs that can be applied to top SCCs analogously.
If a search starting from a vertex of $t_1 \in T_S$ encounters another
vertex $t_2 \in T_S$, $t_1 \ne t_2$, there are two possibilities: either (1)
both vertices are in the same SSC or (2) $t_1$ can reach 
$t_2$ but not vice versa. In Case~(1) the searches from both vertices can 
explore all vertices in the SCC and thus it is sufficient to only search from 
one of them. In Case~(2) the SCC of~$t_1$ has an outgoing
edge and thus cannot be a bottom SCC. Hence in both cases we can remove the vertex
$t_1$ from the set $T_S$ while still maintaining Invariant~\ref{inv:HT}.
By Invariant~\ref{inv:HT} we further have that each search from a vertex of $T_S$
that is not in a bottom SCC encounters another vertex of $T_S$ in its search and 
therefore is removed from the set $T_S$ during Procedure~\ref{proc:lockstep} (if no
top or bottom SCC is found earlier). This ensures that the returned set 
is either a top or a bottom SCC.\footnote{To improve the practical performance, we 
return the updated sets $H_S$ and $T_S$. By the above argument this preserves Invariant~\ref{inv:HT}.}

\vspace{-1mm}
\smallskip\noindent{\em Bound on Symbolic Steps.}
Second, observe that we can still bound the number of symbolic steps needed for 
the search that terminates first by the 
number of \emph{vertices} in the smallest top or bottom SCC of $G[S]$, since this 
is an upper bound on the symbolic steps needed for the search started 
in this SCC. Thus provided Invariant~\ref{inv:HT},
we can bound the number of symbolic steps in Procedure~\ref{proc:lockstep}
to identify a vertex set $\scc \subsetneq S$ such that $\scc$ and $S \setminus \scc$
are not strongly connected in $G[S]$
by $O((\lvert H_S \rvert + \lvert T_S \rvert) \cdot \min(\lvert \scc \rvert, \lvert S \setminus \scc \rvert))$.
In the algorithms that call Procedure~\ref{proc:lockstep} we charge the 
number of symbolic steps in the procedure to the vertices in the smaller 
set of $\scc$ and $S \setminus \scc$; this ensures that each 
vertex is charged at most $O(\log{n})$ times over the whole algorithm.
We obtain the following result (proof in Appendix~\ref{sec:applss}).

\vspace{-1mm}
\begin{thm}[Lock-Step Search]\label{thm:lss}
	Provided Invariant~\ref{inv:HT} holds, 
	Procedure~\ref{proc:lockstep}\upbr{$G$, $S$, $H_S$, $T_S$}
	returns a top or bottom SCC~$\scc$ of $G[S]$. It uses $O((\lvert H_S \rvert 
	+ \lvert T_S \rvert) \cdot \min(\lvert \scc \rvert, \lvert S \setminus \scc \rvert))$  symbolic steps
	if $\scc \ne S$ and $O((\lvert H_S \rvert + \lvert T_S \rvert) \cdot \lvert \scc \rvert)$ otherwise.
\end{thm}

\vspace{-3mm}
\section{Graphs with Streett Objectives}\label{sec:graphs}

\vspace{-1mm}
\smallskip\noindent{\bf Basic Symbolic Algorithm.}
Recall that for a given graph (with $n$ vertices) and a Streett objective (with $k$ 
target pairs) each non-trivial strongly connected subgraph without bad vertices is
a good component. The basic symbolic algorithm for graphs with Streett objectives 
repeatedly removes bad vertices from each SCC and then recomputes the SCCs 
until all good components are found. The winning set then consists of the vertices
that can reach a good component. We refer to this
algorithm as~\ref{alg:streettgraphbasic}. For the pseudocode and more details
see Appendix~\ref{sec:appgraphs}.

\vspace{-0.5mm}
\begin{prp}\label{prp:basicgraphs}
Algorithm~\ref{alg:streettgraphbasic} correctly computes the winning set in
graphs with Streett objectives and requires $O(n \cdot \min(n,k))$ symbolic steps.
\end{prp}

\vspace{-0.5mm}
\smallskip\noindent{\bf Improved Symbolic Algorithm.}
In our improved symbolic algorithm we replace the recomputation of all SCCs
with the search for a new top or bottom SCC with Procedure~\ref{proc:lockstep} 
from vertices that have lost adjacent edges whenever there are not too many such 
vertices. We present the improved symbolic algorithm for graphs with Streett objectives
in more detail as it also conveys important intuition for the MDP case.
The pseudocode is given in Algorithm~\ref{alg:streettgraphimpr}.

\vspace{-0.5mm}
\smallskip\noindent{\em Iterative Refinement of Candidate Sets.}
The improved algorithm maintains a set~$\goodC$ of 
already identified good components that is initially empty
and a set~$\mathcal{\ec}$ of candidates for good components that is initialized with 
the SCCs of the input graph~$G$. The difference to the basic algorithm lies
in the properties of the vertex sets maintained in $\mathcal{\ec}$ and the way we identify 
sets that can be separated from each other without destroying a good component.
In each iteration one vertex set~$S$ is removed from $\mathcal{\ec}$ and, after 
the removal of bad vertices from the set, either identified as a good component or 
split into several candidate sets.
By Lemma~\ref{lem:eccontained} and Corollary~\ref{cor:geccontained}
the following invariant is maintained throughout the algorithm for the sets
 in $\goodC$ and $\mathcal{\ec}$.

\vspace{-0.5mm}
\begin{invariant}[Maintained Sets]\label{inv:gccontained}
The sets in $\mathcal{\ec} \cup \goodC$ are pairwise disjoint and for every good 
component~$\scc$ of $G$ there exists a set $Y \supseteq \scc$
such that either $Y \in \mathcal{\ec}$ or $Y \in \goodC$.
\end{invariant}

\vspace{-0.5mm}
\smallskip\noindent{\em Lost Adjacent Edges.}
In contrast to the basic algorithm, the subgraph induced by a set~$S$ contained 
in~$\mathcal{\ec}$ is not necessarily strongly connected. Instead, we remember
vertices of~$S$ that have lost adjacent edges since the last time a superset
of $S$ was determined to induce a strongly connected subgraph; vertices that lost 
incoming edges are contained in~$H_S$ and vertices that lost 
outgoing edges are contained in~$T_S$. 
In this way we maintain Invariant~\ref{inv:HT} throughout the algorithm,
which enables us to use Procedure~\ref{proc:lockstep} with the running time 
guarantee provided by Theorem~\ref{thm:lss}.

\begin{algorithm2e}[t!]
	\SetAlgoRefName{StreettGraphImpr}
	\caption{Improved Alg. for Graphs with Streett Obj.}
	\label{alg:streettgraphimpr}
	\SetKwInOut{Input}{Input}
	\SetKwInOut{Output}{Output}
	\BlankLine
	\Input{graph $G = (V, E)$ and Streett pairs 
	$\SP= \{(L_i, U_i) \mid 1 \le i \le k\}$
	}
	\Output
	{
	$\win{G, \Streett{\SP}}$
	}
	\BlankLine
	$\mathcal{\ec} \gets \sccalg(G)$; $\goodC \gets \emptyset$\label{l:gimpr:initstart}\;
	\lForEach{$\scc \in \mathcal{\ec}$}{
		$H_\scc \gets \emptyset$; $T_\scc \gets \emptyset$\label{l:gimpr:initend}
	}
	\While{$\mathcal{\ec} \ne \emptyset$\label{l:gimpr:whilestart}}{
		remove some~$S \in \mathcal{\ec}$ from $\mathcal{\ec}$\;
		$\badv \gets \bigcup_{1 \le i \le k : U_i \cap S = \emptyset} (L_i \cap S)$\label{l:gimpr:badstart}\;
		\While{$\badv \ne \emptyset$\label{l:gimpr:innerw}}{
			$S \gets S \setminus \badv$\;
			$H_S \gets (H_S \cup \post(\badv)) \cap S$\;
			$T_S \gets (T_S \cup \pre(\badv)) \cap S$\;
			$\badv \gets \bigcup_{1 \le i \le k : U_i \cap S = \emptyset} (L_i \cap S)$\;
			\label{l:gimpr:badend}
		}
		\If(\tcc*[h]{$G[S]$ contains at least one edge}){$\post(S) \cap S \ne \emptyset$\label{l:p1}}{ 
			\lIf{$\lvert H_S \rvert + \lvert T_S \rvert = 0$}{
				$\goodC \gets \goodC \cup \set{S}$\label{l:gimpr:good1}
			}\ElseIf{$\lvert H_S \rvert + \lvert T_S \rvert \ge \sqrt{m / \log n}$\label{l:gimpr:basicstart}}{
				delete $H_S$ and $T_S$\;
				$\mathcal{\scc} \gets \sccalg(G[S])$\;
				\lIf{$\lvert \mathcal{\scc} \rvert = 1$}{
					$\goodC \gets \goodC \cup \set{S}$\label{l:gimpr:good2}
				}\Else{
					\lForEach{$\scc \in \mathcal{\scc}$}{
						$H_\scc \gets \emptyset$; $T_\scc \gets \emptyset$
					}
                    $\mathcal{\ec} \gets \mathcal{\ec} \cup \mathcal{\scc}$\;
				}
				\label{l:gimpr:basicend}
			}\Else{\label{l:gimpr:lssbeg}
				($\scc$, $H_S$, $T_S$) $\gets $\ref{proc:lockstep}($G$, $S$, $H_S$, $T_S$)\;
				\lIf{$\scc = S$}{
					$\goodC \gets \goodC \cup \set{S}$\label{l:gimpr:good3}
				}\Else(\tcc*[h]{separate $\scc$ and $S \setminus \scc$}){
					$S \gets S \setminus \scc$\;
					$H_\scc \gets \emptyset$; $T_\scc \gets \emptyset$\;
					$H_S \gets (H_S \cup \post(\scc)) \cap S$\label{l:p2}\;
					$T_S \gets (T_S \cup \pre(\scc)) \cap S$\label{l:p3}\;
					$\mathcal{\ec} \gets \mathcal{\ec} \cup \set{S} \cup \set{\scc}$\;
				\label{l:gimpr:lssend}
				}
			}
		}
		\label{l:gimpr:whileend}
	}
	\Return{$\reachG{G}{\bigcup_{\scc \in \goodC} \scc}$}\;
\end{algorithm2e}

\vspace{-0.5mm}
\smallskip\noindent{\em Identifying SCCs.}
Let $S$ be the vertex set removed from $\mathcal{\ec}$ in a fixed iteration of 
Algorithm~\ref{alg:streettgraphimpr} after the removal of bad vertices in 
the inner while-loop. First note that if $S$ is strongly connected and contains
at least one edge, then it is a good component. If the set $S$ was already identified
as strongly connected in a previous iteration, i.e., $H_S$ and $T_S$ are empty,
then $S$ is identified as a good component in line~\ref{l:gimpr:good1}.
If many vertices of $S$ have lost adjacent edges since the last time a super-set
of $S$ was identified as a strongly connected subgraph, then 
the SCCs of $G[S]$ are determined 
as in the basic algorithm. To achieve the optimal 
asymptotic upper bound, we say that many vertices of $S$ have 
lost adjacent edges when we have 
$\lvert H_S \rvert + \lvert T_S \rvert \ge \sqrt{m / \log n}$, while lower 
thresholds are used in our experimental results.
Otherwise, if not too many vertices of~$S$ lost
adjacent edges, then we  start a symbolic \emph{lock-step search} for top
SCCs from the vertices of~$H_S$ and for bottom SCCs from the vertices
of~$T_S$ using Procedure~\ref{proc:lockstep}. 
The set returned by the procedure is either a top or a bottom SCC $\scc$ of
$G[S]$ (Theorem~\ref{thm:lss}). Therefore we can from now on consider $\scc$
and $S \setminus \scc$ separately, maintaining
Invariants~\ref{inv:HT} and~\ref{inv:gccontained}. 

\vspace{-0.5mm}
\smallskip\noindent{\em Algorithm~\ref{alg:streettgraphimpr}.} A succinct description of the 
pseudocode is as follows: 
Lines~\ref{l:gimpr:initstart}--\ref{l:gimpr:initend} initialize the set of 
candidates for good components with the SCCs of the input graph.
In each iteration of the main while-loop one candidate is considered and the following
operations are performed:
(a)~lines~\ref{l:gimpr:badstart}--\ref{l:gimpr:badend} iteratively remove all bad vertices; if afterwards
the candidate is still strongly connected (and contains at least one edge), it is identified as a good component
in the next step; otherwise it is partitioned into new candidates in one of the following ways:
(b)~if many vertices lost adjacent edges, lines~\ref{l:gimpr:basicstart}--\ref{l:gimpr:basicend} partition the
candidate into its SCCs (this corresponds to an iteration of the basic algorithm);
(c)~otherwise, lines~\ref{l:gimpr:lssbeg}--\ref{l:gimpr:lssend} use symbolic 
lock-step search to partition the candidate into one of its SCCs and the remaining vertices.
The while-loop terminates when no candidates are left.
Finally, vertices that can reach some good component are returned.
We have the following result (proof in Appendix~\ref{sec:appgraphs}).

\vspace{-0.5mm}
\begin{thm}[Improved Algorithm for Graphs]\label{thm:improvedgraphs}
Algorithm~\ref{alg:streettgraphimpr} correctly
computes the winning set in graphs with Streett objectives and requires
$O(n \cdot \sqrt{m \log n})$ symbolic steps.
\end{thm}

\vspace{-3mm}
\section{Symbolic MEC Decomposition}\label{sec:mecs}

\vspace{-1mm}
In this section we present a succinct description of the basic symbolic
algorithm for MEC decomposition and then present the main ideas
for the improved algorithm.

\vspace{-0.5mm}
\smallskip\noindent{\em Basic symbolic algorithm for MEC decomposition.} 
The basic symbolic algorithm for MEC decomposition maintains a set of identified
MECs and a set of candidates for MECs, initialized with the SCCs of the MDP.
Whenever a candidate is considered, either
(a)~it is identified as a MEC or
(b)~it contains vertices with outgoing random edges, which are then removed 
together with their random attractor from the candidate, and the SCCs
of the remaining sub-MDP are added to the set of candidates.
We refer to the algorithm as~\ref{alg:mecbasic}.

\vspace{-0.5mm}
\begin{prp}\label{prp:basicmecs}
Algorithm~\ref{alg:mecbasic} correctly computes the MEC decomposition of MDPs and requires $O(n^2)$ symbolic steps.
\end{prp}

\vspace{-0.5mm}
\smallskip\noindent{\em Improved symbolic algorithm for MEC decomposition.}
The improved symbolic algorithm for MEC decomposition uses the ideas of symbolic
lock-step search presented in Section~\ref{sec:lss}.
Informally, when considering a candidate that lost a few edges from the remaining
graph, we use the symbolic lock-step search to identify some bottom SCC.
We refer to the algorithm as~\ref{alg:mecimpr}.
Since all the important conceptual ideas regarding the symbolic lock-step search 
are described in Section~\ref{sec:lss}, we relegate the technical details to Appendix~\ref{sec:appmecs}.
We summarize the main result (proof in Appendix~\ref{sec:appmecs}).

\vspace{-0.5mm}
\begin{thm}[Improved Algorithm for MEC]\label{thm:improvedmecs}
Algorithm~\ref{alg:mecimpr} correctly computes the MEC decomposition of MDPs and requires $O(n\cdot \sqrt{m})$ symbolic steps.
\end{thm}

\section{MDPs with Streett Objectives}\label{sec:mdps}

\smallskip\noindent{\bf Basic Symbolic Algorithm.}
We refer to the basic symbolic algorithm for MDPs with Streett objectives
as~\ref{alg:streettmdpbasic}, which is similar to the algorithm for graphs,
with SCC computation replaced by MEC computation.
The pseudocode of Algorithm~\ref{alg:streettmdpbasic} together with its
detailed description is presented in Appendix~\ref{sec:appmdps}. 

\begin{prp}\label{prp:basicmdps}
Algorithm~\ref{alg:streettmdpbasic} correctly computes the almost-sure winning set
in MDPs with Streett objectives and requires $O(n^2 \cdot \min(n,k))$ symbolic steps.
\end{prp}

\begin{rmk}
The above bound uses the basic symbolic MEC decomposition algorithm. 
Using our improved symbolic MEC decomposition algorithm, the above bound
could be improved to 
$O(n \cdot \sqrt{m} \cdot \min(n,k))$.
\end{rmk}

\smallskip\noindent{\bf Improved Symbolic Algorithm.}
We refer to the improved symbolic algorithm for MDPs with Streett objectives as~\ref{alg:streettmdpimpr}.
First we present the main ideas for the improved symbolic algorithm. Then we explain
the key differences compared to the improved symbolic algorithm for graphs. A thorough description
with the technical details and proofs is presented in Appendix~\ref{sec:appmdps}.
\begin{compactitem}
\item First, we improve the algorithm by interleaving the symbolic MEC computation with the detection
of bad vertices~\cite{ChatterjeeDHL16,Loitzenbauer16}. This allows to replace the
computation of MECs in each iteration of the while-loop with the
computation of SCCs and an additional random attractor computation.
\begin{compactitem}
\item {\em Intuition of interleaved computation.}
Consider a candidate for a good end-component $S$ after a random attractor
to some bad vertices is removed from it. After the removal of the random attractor,
the set $S$ does not have random vertices with outgoing edges. Consider that further $\bad(S) = \emptyset$ holds.
If $S$ is strongly connected and contains an edge, then it is a good end-component.
If $S$ is not strongly connected, then $\mdp[S]$ contains at least two SCCs and some
of them might have random vertices with outgoing edges. Since end-components are strongly connected and
do not have random vertices with outgoing edges, we have that (1) every good end-component is completely
contained in one of the SCCs of $\mdp[S]$ and (2) the random vertices of an SCC with outgoing edges
and their random attractor do not intersect with any good end-component (see Lemma~\ref{lem:eccontained}).
\item {\em Modification from basic to improved algorithm.}
We use these observations to modify the basic algorithm as follows:
First, for the sets that are candidates for good end-components, we do not maintain the property
that they are end-components, but only that they do not have random vertices with outgoing edges
(it still holds that every maximal good end-component is either already identified or 
contained in one of the candidate sets). Second, for a candidate set $S$, we repeat the removal of 
bad vertices until $\bad(S) = \emptyset$ holds before we continue with the 
next step of the algorithm. This allows us to make progress after the removal of bad vertices
by computing all SCCs (instead of MECs) of the remaining sub-MDP.
If there is only one SCC, then this is a good end-component (if it contains 
at least one edge). Otherwise (a) we remove from each SCC
the set of random vertices with outgoing edges and their random attractor
and (b) add the remaining vertices of each SCC as a new candidate set.

\end{compactitem}
\item Second, as for the improved symbolic algorithm for graphs, we use 
the symbolic lock-step search to quickly identify a top or bottom SCC every time
a candidate has lost a small number of edges since the last time its superset
was identified as being strongly connected. The symbolic lock-step search
is described in detail in Section~\ref{sec:lss}.
\end{compactitem}

\vspace{2mm}
Using interleaved MEC computation and lock-step search leads to a similar algorithmic
structure for Algorithm~\ref{alg:streettmdpimpr} as for our improved symbolic algorithm for graphs
(Algorithm~\ref{alg:streettgraphimpr}). The key differences are as follows:
First, the set of candidates for good end-components is initialized with the MECs of
the input graph instead of the SCCs. Second, whenever bad vertices are removed
from a candidate, also their random attractor is removed.
Further, whenever a candidate is partitioned into its SCCs, for each SCC, the
random attractor of the vertices with outgoing random edges is removed.
Finally, whenever a candidate $S$ is separated into
$C$ and $S \setminus C$ via symbolic lock-step search, the random attractor of
the vertices with outgoing random edges is removed from $C$, and the random
attractor of $C$ is removed from $S$.

\begin{thm}[Improved Algorithm for MDPs]\label{thm:improvedmdps}
Algorithm~\ref{alg:streettmdpimpr} correctly computes the almost-sure winning set
in MDPs with Streett objectives and requires \mbox{$O(n \cdot \sqrt{m \log n})$}
symbolic steps.
\end{thm}

\section{Experiments}\label{sec:exper}

We present a basic prototype implementation of our algorithm and compare against 
the basic symbolic algorithm for graphs and MDPs with Streett objectives.

\smallskip\noindent{\em Models.}
We consider the academic benchmarks from the VLTS benchmark suite~\cite{VLTS}, which gives
representative examples of systems with nondeterminism, and has been used in previous
experimental evaluation (such as~\cite{BarnatCP11,ChatterjeeHJS13}).

\smallskip\noindent{\em Specifications.}
We consider random LTL formulae and use the tool Rabinizer~\cite{KomarkovaK14} to 
obtain deterministic Rabin automata. Then the negations of the formulae
give us Streett automata, which we consider as the specifications.

\smallskip\noindent{\em Graphs.}
For the models of the academic benchmarks, we first compute SCCs, as all 
algorithms for Streett objectives compute SCCs as a preprocessing step. 
For SCCs of the model benchmarks we consider products with the specification Streett
automata, to obtain graphs with Streett objectives, which are the benchmark
examples for our experimental evaluation. The number of transitions in the
benchmarks ranges from $300$K to $5$Million.

\smallskip\noindent{\em MDPs.}
For MDPs, we consider the graphs obtained as above and consider a 
fraction of the vertices of the graph as random vertices, which is 
chosen uniformly at random. 
We consider $10\%$, $20\%$, and $50\%$ of the vertices as random vertices
for different experimental evaluation.

\setlength{\abovecaptionskip}{1pt}
\begin{figure}[t]
\centering
\includegraphics[width=0.5\textwidth]{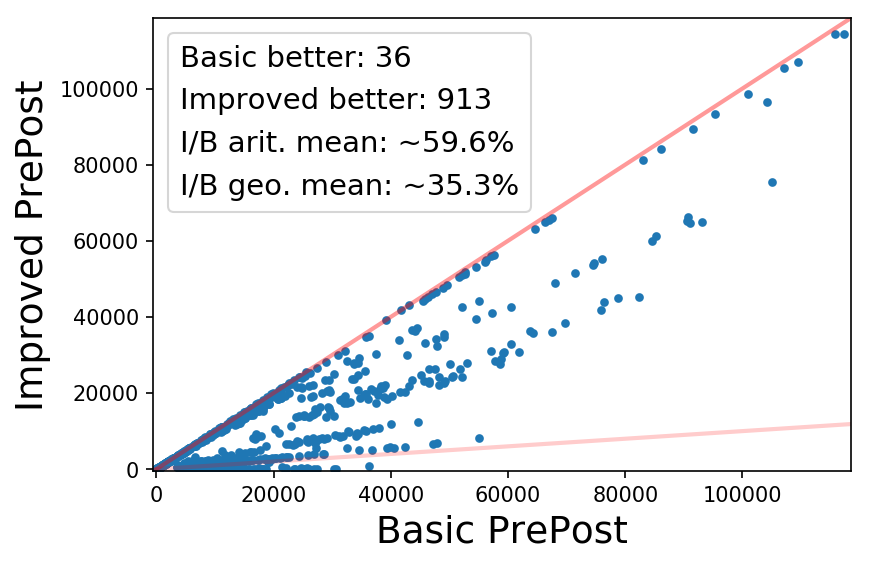}
\caption{Results for graphs with Streett objectives.}
\label{fig:graphs}
\end{figure}
\setlength{\abovecaptionskip}{6pt}
\begin{figure}[H]
\begin{center}
	\vspace{-2mm}
     \subfloat[10\% random vertices\label{fig:mdp10}]{
       \includegraphics[width=0.5\textwidth]{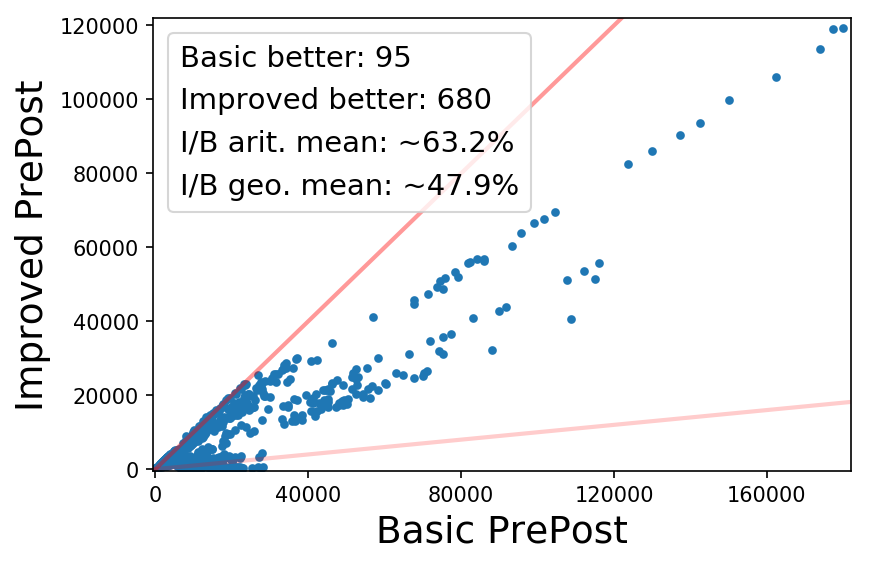}
     }
     \subfloat[20\% random vertices\label{fig:mdp20}]{
       \includegraphics[width=0.5\textwidth]{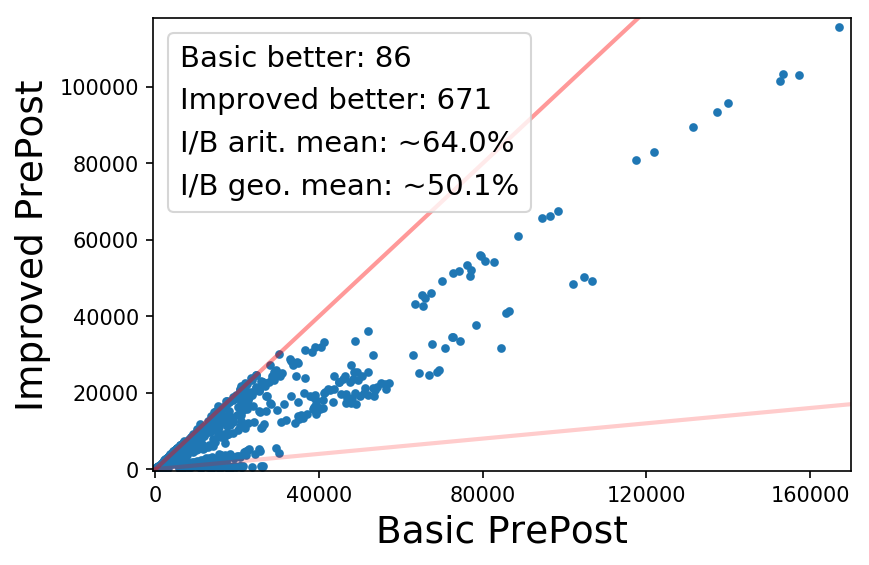}
     }
     \hfill
	\vspace{-2mm}
     \subfloat[50\% random vertices\label{fig:mdp50}]{
       \includegraphics[width=0.5\textwidth]{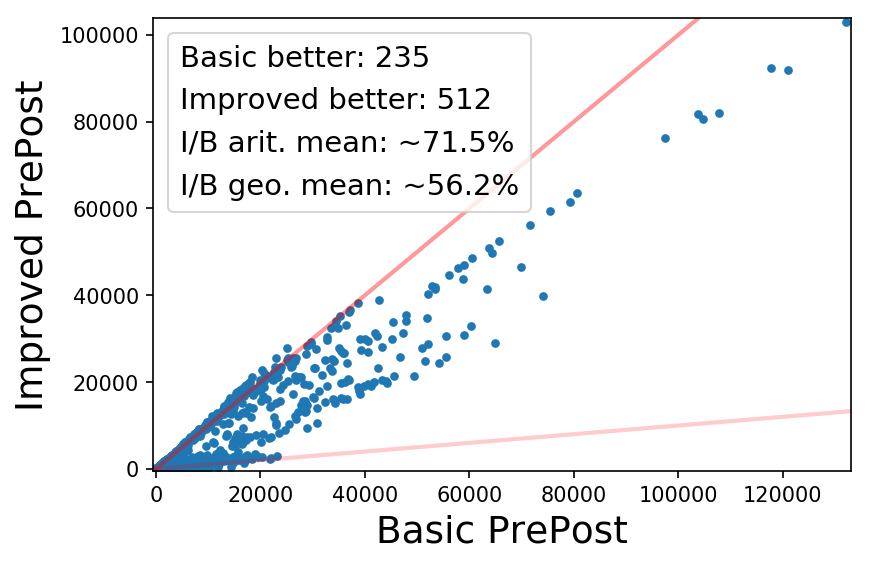}
     }
\caption{Results for MDPs with Streett objectives.}
\label{fig:mdps}
\end{center}
\end{figure}

\smallskip\noindent{\em Experimental evaluation.}
In the experimental evaluation we compare the number of symbolic steps 
(i.e., the number of $\pre/\post$ operations\footnote{Recall that the basic
set operations are cheaper to compute, and asymptotically at most the
number of $\pre/\post$ operations in all the presented algorithms.})
executed by the algorithms, the comparison of running time yields similar
results and is provided in Appendix~\ref{sec:appexper}.
As the initial preprocessing step is the same for all the algorithms
(computing all SCCs for graphs and all MECs for MDPs), the comparison
presents the number of symbolic steps executed after the preprocessing.
The experimental results for graphs are shown in Figure~\ref{fig:graphs} and
the experimental results for MDPs are shown in Figure~\ref{fig:mdps} (in each
figure the two lines represent equality and an order-of-magnitude improvement, respectively).

\smallskip\noindent{\em Discussion.}
Note that the lock-step search is the key reason for theoretical 
improvement, however, the improvement relies on a large number of Streett 
pairs. In the experimental evaluation, the LTL formulae generate
Streett automata with small number of pairs, which after the product
with the model accounts for an even smaller fraction of pairs as compared to
the size of the state space. This has two effects:
\begin{compactitem}
\item In the experiments the lock-step search is performed for a much smaller 
parameter value ($O(\log n)$ instead of the theoretically optimal  bound of 
$\sqrt{m/\log n}$), and leads to a small improvement.
\item For large graphs, since the number of pairs is small as compared to the 
number of states, the improvement over the basic algorithm is minimal.
\end{compactitem}
In contrast to graphs, in MDPs even with small number of pairs as 
compared to the state-space, the interleaved MEC computation has
a notable effect on practical performance, and we observe performance improvement 
even in large MDPs.

\section{Conclusion}\label{sec:conc}

In this work we consider symbolic algorithms for graphs and MDPs with 
Streett objectives, as well as for MEC decomposition. 
Our algorithmic bounds match for both graphs and MDPs. 
In contrast, while SCCs can be computed in linearly many symbolic steps no such
algorithm is known for MEC decomposition.
An interesting direction of future work would be to explore further 
improved symbolic algorithms for MEC decomposition.
Moreover, further improved symbolic algorithms for graphs and MDPs with Streett
objectives is also an interesting direction of future work.

\subsubsection*{Acknowledgements.}
K.~C. and M.~H.\ are partially supported by
the Vienna Science and Technology Fund (WWTF) grant ICT15-003.
K.~C.\ is partially supported by
the Austrian Science Fund (FWF): S11407-N23 (RiSE/SHiNE), and an ERC Start Grant (279307: Graph Games).
V.~T. \ is partially supported by
the European Union’s Horizon 2020 research and innovation programme under the Marie Skłodowska-Curie Grant Agreement No. 665385.
V.~L.~is partially supported by
the Austrian Science Fund (FWF): S11408-N23 (RiSE/SHiNE), the ISF grant \#1278/16, and an ERC Consolidator Grant (project MPM).
For M.~H.\ and V.~L.\ the research leading to these results has received funding from
the European Research Council under the European Union's Seventh Framework Programme (FP/2007-2013) / ERC Grant Agreement no. 340506.

\bibliographystyle{splncs03}
\bibliography{literature}

\clearpage
\section*{Appendix}
\appendix
\section{Details of Section~\ref{sec:lss}: Symbolic Lock-Step Search}\label{sec:applss}

\begin{proof}[of Theorem~\ref{thm:lss}]

\smallskip\noindent{\bf Strong connectivity.}
We want to show that $C \gets $ \ref{proc:lockstep}($G$, $S$, $H_S$, $T_S$) is
a top or bottom SCC of $G[S]$ given Invariant~\ref{inv:HT} is satisfied. 
By the invariant at least one vertex of each top SCC of $G[S]$ is contained in~$H_S$
and at least one vertex of each bottom SCC of $G[S]$ is contained in~$T_S$.
Suppose $C$ is the set obtained from a search conducted by $\post$ operations 
that started from within a
bottom SCC $\tilde{C}$ of $G[S]$. Since $\tilde{C}$ is a bottom SCC
and we update the search by executing $\post$ operations (and moreover
intersect with $S$ at every update), we have $C \subseteq \tilde{C}$. Further,
since $\tilde{C}$ is an SCC, the updates
with $\post$ eventually cover all vertices of $\tilde{C}$, which gives
us $C = \tilde{C}$. A set~$C_t$ constructed with $\post$ operations whose start vertex~$t$
is not contained in a bottom SCC of $G[S]$ can not yield the set $C$ since
eventually it contains a bottom SCC of $G[S]$, and by
Invariant~\ref{inv:HT} this SCC contains a candidate in $T_S$; therefore 
$|C_t \cap T_S| > 1$ is satisfied at some point in the construction of $C_t$ and
then search is canceled by removing $t$ from $T_S$; note that a search starting 
from a bottom SCC can be canceled only if another vertex of the bottom SCC remains in $T_S$. By the
symmetric argument for searches conducted by $\pre$ operations that
started from a vertex of a top SCC we have that the returned set $C$ is either 
a top or a bottom SCC of $G[S]$.

\smallskip\noindent{\bf Bound on symbolic steps.}
Consider (one of) the smallest top or bottom SCCs $\tilde{C}$ of $G[S]$. Suppose
w.l.o.g.\ that $\tilde{C}$ is a bottom SCC. By
Invariant~\ref{inv:HT} there is a search, conducted by $\post$ operations, 
that starts from a vertex $t \in T_S$
within $\tilde{C}$ and that is not canceled, and
therefore this search terminates after at most $\lvert \tilde{C} \rvert$ many 
$\post$ operations. Other searches may terminate earlier but this gives an upper bound 
of $O((\lvert H_S \rvert + \lvert T_S \rvert) \cdot \lvert \tilde{C} \rvert)$ on 
the number of symbolic steps until the lock-step search terminates. 
Finally, consider the returned set
$C \gets $ \ref{proc:lockstep}($G$, $S$, $H_S$, $T_S$). 
There are two possible cases: either (i) $S = C$, which implies $C = \tilde{C}$ so the number of
symbolic steps can be bounded by $O((\lvert H_S \rvert + \lvert T_S \rvert) \cdot \lvert C \rvert)$, 
or (ii) $S \neq C$. In the second case, since $\tilde{C}$ is (some) smallest SCC, $C$ is an SCC,
and $S \setminus C$ contains at least one SCC, we have $\lvert \tilde{C} \rvert \le \lvert C \rvert$
and $\lvert \tilde{C} \rvert \le \lvert S \setminus C \rvert$, and hence we can bound the number
of symbolic steps in this case by $O((\lvert H_S \rvert + \lvert T_S \rvert) \cdot 
\min(\lvert \scc \rvert, \lvert S \setminus \scc \rvert))$.
\qed
\end{proof}

\section{Details of Section~\ref{sec:graphs}: Graphs with Streett Objectives}\label{sec:appgraphs}

\subsection{Basic Symbolic Algorithm for Graphs with Streett Objectives}

\begin{algorithm2e}
	\SetAlgoRefName{StreettGraphBasic}
	\caption{Basic Algorithm for Graphs with Streett Obj.}
	\label{alg:streettgraphbasic}
	\SetKwInOut{Input}{Input}
	\SetKwInOut{Output}{Output}
	\BlankLine
	\Input{graph $G = (V, E)$ and Streett pairs 
	$\SP= \{(L_i, U_i) \mid 1 \le i \le k\}$
	}
	\Output
	{
	$\win{G, \Streett{\SP}}$
	}
	\BlankLine
	$\mathcal{\ec} \gets \sccalg(G)$; $\goodC \gets \emptyset$\;
	\While{$\mathcal{\ec} \ne \emptyset$}{
		remove some~$S \in \mathcal{\ec}$ from $\mathcal{\ec}$\;
		$\badv \gets \bigcup_{1 \le i \le k : U_i \cap S = \emptyset} (L_i \cap S)$\;
		\If{$\badv \ne \emptyset$}{
			$S \gets S \setminus \badv$\;
			$\mathcal{\ec} \gets \mathcal{\ec} \cup \sccalg(G[S])$\;
		}\Else{
			\If(\tcc*[h]{$G[S]$ contains at least one edge}\label{l:gbasic:edge})
			{$\post(S) \cap S \ne \emptyset$}{
				$\goodC \gets \goodC \cup \set{S}$\;
			}
		}
	}
	\Return{$\reachG{G}{\bigcup_{\scc \in \goodC} \scc}$}\;
\end{algorithm2e}

The pseudocode of the basic symbolic algorithm for graphs with Streett
objectives is given in Algorithm~\ref{alg:streettgraphbasic}.

The basic symbolic algorithm for Streett objectives on graphs finds good components as follows.
The algorithm maintains two sets of vertex sets: $\goodC$ contains identified 
good components and is initially empty; $\mathcal{\ec}$ contains candidates for 
good components and is initialized with the SCCs of the input graph~$G$. The sets 
in $\mathcal{\ec}$ are strongly connected subgraphs of~$G$ throughout the 
algorithm. In each iteration of the while-loop one of the candidate sets~$S$
maintained in $\mathcal{\ec}$ is considered. If the set~$S$ does not contain
bad vertices and contains at least one edge, then it is a good component and 
added to $\goodC$. Otherwise, the set of bad vertices~$\badv$ in~$S$ is removed from $S$; the
subgraph induced by $S' = S \setminus \badv$ might not be strongly connected
but every good component contained in $S'$ must still be strongly connected,
therefore the maximal strongly connected subgraphs of $G[S']$ are added 
to $\mathcal{\ec}$ as new candidates for good components. By Lemma~\ref{lem:eccontained} 
and Corollary~\ref{cor:geccontained} this procedure maintains the property that 
every good component of~$G$ is completely contained in one of the vertex sets of
$\goodC$ or $\mathcal{\ec}$. Further in each iteration either (a) 
vertices are removed or separated into different vertex sets or (b) a new good 
component is identified. Thus after at most $O(n)$ iterations the set $\mathcal{\ec}$
is empty and all good components of~$G$ are contained in $\goodC$.
Furthermore, whenever bad vertices are removed from a given candidate set,
the number of target pairs this candidate set intersects is reduced by one.
Thus each vertex is considered in at most $O(k)$ iterations of the main while-loop.
Finally, the set of vertices that can reach a good component is determined (by $O(n)$ 
$\pre$ operations) and output as the winning set.
Since computing SCCs can be done in $O(n)$ symbolic steps, the total number of
symbolic steps of the basic algorithm is bounded by $O(n \cdot \min(n, k))$.

\subsection{Improved Symbolic Algorithm for Graphs with Streett Objectives}

\begin{lemma}[Invariants of Improved Algorithm for Graphs]\label{lem:improvedgraphsinv}
Invariant~\ref{inv:HT} and Invariant~\ref{inv:gccontained} are preserved throughout
Algorithm~\ref{alg:streettgraphimpr}, i.e., they hold before the first iteration, after each
iteration, and after termination of the main while-loop. Further, Invariant~\ref{inv:HT}
is preserved during each iteration of the main while-loop.
\end{lemma}

\begin{proof}
\item
\smallskip\noindent{\bf Invariant~\ref{inv:HT}.}
Whenever a new candidate $S$ is added as a result from $\sccalg$, it is strongly
connected, and we set $H_S = T_S = \emptyset$; this in particular implies that 
the invariant is satisfied after the initialization of the algorithm. 

By induction and Theorem~\ref{thm:lss}, the invariant is satisfied whenever
Procedure~\ref{proc:lockstep} returns a candidate $C$ and we set
$H_C = T_C = \emptyset$.

Now consider an update of a candidate $S$ where some subset
$B$ is deleted from it and assume the invariant holds before the update. In these
cases we update $H_S$ and $T_S$ by setting $H_S \gets (H_S \cup \post(B)) \cap S$ 
and $T_S \gets (T_S \cup \pre(B)) \cap S$. This adds the vertices that remain 
in~$S$ and have an edge from a vertex of $B$ to $H_S$ and those with an edge 
to $B$ to $T_S$. Suppose a new
top (resp. bottom) SCC $\tilde{S} \subseteq S$ emerges in $S$ by the removal 
of $B$ from $S$. Then some vertex of $\tilde{S}$ had an outgoing edge to $B$ (resp.\ an incoming edge from $B$)
and thus is contained in the updated set $T_S$ (resp.\ $H_S$),
maintaining the invariant. 
This happens whenever we remove $\bad(S)$ from $S$, and whenever
we subtract a result from Procedure~$\ref{proc:lockstep}$ $C$ from $S$.

\smallskip\noindent{\bf Invariant~\ref{inv:gccontained} -- Disjointness.}
The sets in $\mathcal{\ec} \cup \goodC$ are pairwise disjoint at the
initialization since $\goodC$ is initialized as $\emptyset$. Furthermore,
whenever a set $S$ is added to $\goodC$ in an iteration of the main
while-loop, a superset $\tilde{S} \supseteq S$ is removed from
$\mathcal{\ec}$ in the same iteration of the while-loop. Therefore
by induction the disjointness of the sets in $\mathcal{\ec} \cup \goodC$
is preserved.

\smallskip\noindent{\bf Invariant~\ref{inv:gccontained} -- Containment of good components.}
At initialization, $\mathcal{\ec}$ contains all SCCs of the input graph $G$.
Each good component $C$ of $G$ is strongly connected, so there exists an SCC
$Y \supseteq C$ such that $Y \in \mathcal{\ec}$ for each good component $C$.

Consider a set $S \in \mathcal{\ec}$ that is removed from $\mathcal{\ec}$
at the beginning of an iteration of the main while-loop. Consider further
a good component $C$ of $G$ such that $C \subseteq S$. We require that a set
$Y \supseteq C$ is added to either
$\mathcal{\ec}$ or $\goodC$ in this iteration of the main while-loop.

First, whenever we remove $\bad(S)$ from $S$, by Corollary~\ref{cor:geccontained}
we maintain the fact that $C \subseteq S$. Second, $G[S]$ contains an edge since $C \subseteq S$.
Finally, one of the three cases happens:

\smallskip\noindent Case~(1): If $\lvert H_S \rvert + \lvert T_S \rvert = 0$, then the set
$S \supseteq C$ is added to $\goodC$.

\smallskip\noindent Case~(2): If $\lvert H_S \rvert + \lvert T_S \rvert \ge \sqrt{m / \log n}$,
then the algorithm computes the SCCs of $G[S]$. Since $C \subseteq S$ is strongly connected,
it is completely contained in some SCC $Y$ of $G[S]$, and $Y$ is added either to
$\mathcal{\ec}$ or to $\goodC$.

\smallskip\noindent Case~(3): If $0 < \lvert H_S \rvert + \lvert T_S \rvert < \sqrt{m / \log n}$,
then the algorithm either adds $S \supseteq C$ to $\goodC$, or partitions $S$
into $\tilde{S}$ and $S \setminus \tilde{S}$. Suppose the latter case happens, then by
Theorem~\ref{thm:lss} we have that $\tilde{S}$ is an SCC of $G[S]$. Further, since $C \subseteq S$
is strongly connected, it is completely contained in some SCC of $G[S]$.
Therefore either $C \subseteq \tilde{S}$ or $C \subseteq (S \setminus \tilde{S})$, and
both $\tilde{S}$ and $S \setminus \tilde{S}$ are added to $\mathcal{\ec}$.

\smallskip\noindent By the above case analysis we have that a set $Y \supseteq C$ is added
to either $\mathcal{\ec}$ or $\goodC$ in the iteration of the main while-loop, and thus
the invariant is preserved throughout the algorithm.
\qed

\end{proof}

\begin{proof}[of Theorem~\ref{thm:improvedgraphs}]

\smallskip\noindent{\bf Correctness.}
	Whenever a candidate set $S$ is added to $\goodC$, it contains an edge
    by the check at line~\ref{l:p1}, and $\bad(S) = \emptyset$ by the check at
    line~\ref{l:gimpr:innerw}. Furthermore, \upbr{a} at line~\ref{l:gimpr:good1}, $S$
    is strongly connected by Invariant~\ref{inv:HT}, \upbr{b} at line~\ref{l:gimpr:good2},
    $S$ is strongly connected by the result of~\sccalg, and \upbr{c} at
    line~\ref{l:gimpr:good3}, $S$ is strongly connected by Theorem~\ref{thm:lss}.
    Therefore we have that whenever a candidate set is added to $\goodC$,
    it is indeed a good component (soundness).

    Finally, by soundness, Invariant~\ref{inv:gccontained}, the termination of 
    the algorithm (shown below), and the fact that $\mathcal{\ec} = \emptyset$ at the
    termination of the algorithm, we have that $\goodC$ contains all good
    components of $G$ (completeness).

\smallskip\noindent{\bf Symbolic steps analysis.}
	By~\cite{GentiliniPP08}, the initialization with the SCCs of
	the input graph takes $O(n)$ symbolic steps.
	Furthermore, the reachability computation in the last step
	takes $O(n)$ $\pre$ operations.
	
	In each iteration of the outer while-loop, a set~$S$ is removed 
	from $\mathcal{\ec}$ and either \upbr{a} a set $S' \subseteq S$ is added
    to $\goodC$ and no set is added to $\mathcal{\ec}$ or \upbr{b} at least two
    sets that are (proper subsets of) a partition of $S$ are added to $\mathcal{\ec}$.
    Both can happen at most $O(n)$ times,
	thus there can be at most $O(n)$ iterations of the outer while-loop.
	The $\pre$ and $\post$ operations at lines~\ref{l:p1},~\ref{l:p2}, 
    and~\ref{l:p3} can be charged to the iterations of the outer while-loop.
    
	An iteration of the inner while-loop (lines~\ref{l:gimpr:innerw}-\ref{l:gimpr:badend})
	is executed only if some vertices~$\badv$
	are removed from $S$; the vertices of $\badv$
	are then not considered further. Thus there can, in total, be at most $O(n)$
	$\pre$ and $\post$ operations over all iterations of the inner while-loop.
	
    Note that every vertex in each
    of $H_S$ and $T_S$ can be attributed to at least one unique implicit edge deletion 
    since we only add vertices to $H_S$ resp.\ $T_S$ that are successors resp.\ 
    predecessors of vertices that were separated from $S$ 
    (or deleted from the maintained graph).
    Whenever the case $\lvert H_S \rvert + \lvert T_S \rvert \ge \sqrt{m / \log{n}}$
    occurs, for all subsets $C \subseteq S$ that are then added to
    $\mathcal{\ec}$, we initialize $H_C = T_C = \emptyset$. Therefore the
    case $\lvert H_S \rvert + \lvert T_S \rvert \ge \sqrt{m / \log{n}}$ can
    happen at most $O(\sqrt{m \log{n}})$ times throughout the algorithm since
    there are at most $m$ edges that can be deleted, and hence in total
    takes $O(n \cdot \sqrt{m \log{n}})$ symbolic steps.
	
	It remains to bound the number of symbolic steps in
	Procedure~\ref{proc:lockstep}. Let $\scc$ be the set returned by the procedure;
	we charge the symbolic steps in this call of the procedure to the vertices of 
	the smaller set of $\scc$ and $S \setminus \scc$.
	By Theorem~\ref{thm:lss} we have either \upbr{a}
	$\scc = S$, the number of symbolic steps in this call is 
	bounded by $O(\sqrt{m/\log{n}} \cdot \lvert \scc \rvert)$, and the set $S$ is added to $\goodC$ or
	\upbr{b} $\min(\lvert \scc \rvert, \lvert S \setminus 
	\scc \rvert) \le \lvert S \rvert / 2$ and the
	number of symbolic steps in this call is 
	bounded by $O(\sqrt{m/\log{n}} \cdot \min(\lvert \scc \rvert, \lvert S \setminus 
	\scc \rvert))$.
	Case~\upbr{a} can happen at most once for the vertices of $\scc$,
	and for case~\upbr{b} note that the size of a set containing a specific vertex 
	can be halved at most $O(\log{n})$ times; thus we charge each vertex at 
	most $O(\log{n})$ times. Hence we can bound the total number of symbolic
	steps in all calls to the procedure by $O(n \cdot \sqrt{m \log{n}})$.
    \qed

\end{proof}

\section{Details of Section~\ref{sec:mecs}: Symbolic MEC Decomposition}\label{sec:appmecs}

\subsection{Basic Symbolic Algorithm for MEC decomposition}
Recall that an end-component is a set of vertices that (a) has no 
random edges to vertices not in the set and its induced sub-MDP
is (b) strongly connected and (c) contains at least one edge.

Algorithm~\ref{alg:mecbasic} computes all maximal end-components of
a given MDP and is formulated as to highlight the similarities to the algorithms
for graphs and MDPs with Streett objectives.
The algorithm maintains two sets, the set $\goodC$ of identified maximal 
end-components that is initially empty and the set $\mathcal{\ec}$ of candidates 
for maximal end-components that is initialized with the SCCs of 
the MDP. In each iteration of the while-loop one set $S$ is removed from 
$\mathcal{\ec}$ and either (1a) identified as a maximal end-component and added 
to $\goodC$ or (1b) removed because the induced sub-MDP does not contain an edge
or (2) it contains vertices with outgoing random edges. In the latter case 
these vertices $\rout$ are identified and their random attractor is removed from 
$S$. After this step the sub-MDP induced by the remaining vertices 
of $S$ might not be strongly connected any more. Therefore the SCCs of this 
sub-MDP are determined and added to $\mathcal{\ec}$ as new candidates 
for maximal end-components. Note that this maintains the invariants that 
(i) each set in $\mathcal{\ec}$ induces a strongly connected subgraph and 
(ii) each end-component is a subset of one set in either 
$\goodC$ or $\mathcal{\ec}$. By (i) a set in $\mathcal{\ec}$ is an end-component
if it does not have outgoing random edges and the induced sub-MDP contains an edge, 
i.e., in particular this holds for the sets added to $\goodC$ (soundness). By (ii) and 
$\mathcal{\ec} = \emptyset$ at termination of the while-loop
the algorithm identifies all maximal end-components of the MDP (completeness).
Since both (1) and (2) can happen at most $O(n)$ times, there are $O(n)$
iterations of the while-loop. 
In each iteration the most expensive operations are the computation
of a random attractor and of SCCs, which can both be done in $O(n)$ symbolic steps.
Thus Algorithm~\ref{alg:mecbasic} correctly computes all maximal end-components
of an MDP and takes $O(n^2)$ symbolic steps.

\begin{algorithm2e}[t]
	\SetAlgoRefName{MECBasic}
	\caption{Basic Algorithm for Maximal End-Components}
	\label{alg:mecbasic}
	\SetKwInOut{Input}{Input}
	\SetKwInOut{Output}{Output}
	\BlankLine
	\Input{an MDP~$\mdp = (G = (V, E), (\vp, \vr))$
	}
	\Output
	{
	 the set of maximal end-components of $\mdp$
	}
	\BlankLine
	$\goodC \gets \emptyset$\;
	$\mathcal{\ec} \gets \sccalg(G)$\;
	\While{$\mathcal{\ec} \ne \emptyset$}{
		remove some $S \in \mathcal{\ec}$ from $\mathcal{\ec}$\;
		$\rout \gets S \cap \vr \cap \pre(V\setminus S) $\;
		\If{$\rout \ne \emptyset$}{
			$S \gets S \setminus \at{R}{G}{\rout}$\;
			$\mathcal{\ec} \gets \mathcal{\ec} \cup \sccalg(G[S])$\;
		}\Else{
			\If(\tcc*[h]{$G[S]$ contains at least one edge})
			{$\post(S) \cap S \ne \emptyset$}{
				$\goodC \gets \goodC \cup \set{S}$\;
			}
		}
	}
	\Return{$\goodC$}\;
\end{algorithm2e}

\subsection{Improved Symbolic Algorithm for MEC decomposition}

\smallskip\noindent{\em Informal description.}
We show how to determine all maximal end-components (MECs) of an MDP 
in $O(n \sqrt{m})$ symbolic operations. The difference to the basic algorithm  lies in the way strongly connected parts
of the MDP are identified after the deletion of vertices that cannot be contained
in a MEC. For this the symbolic lock-step search from Section~\ref{sec:lss} is used whenever
not too many edges have been deleted since the last re-computation of SCCs.

Let $\mdp$ be the given MDP and $G = (V, E)$ its underlying graph.
The algorithm maintains two sets of vertex sets: the set $\goodC$ of already
identified MECs that is initialized with the empty set and the set $\mathcal{\ec}$
that is initialized with the SCCs of $G$ and contains vertex sets that are
candidates for MECs. The algorithm preserves the 
following invariant for the $\goodC$ and $\mathcal{\ec}$ over the iterations 
of the while-loop and returns the set $\goodC$
when the set $\mathcal{\ec}$ is empty after an iteration of the while-loop.

\begin{invariant}[Maintained Sets]\label{inv:meccontained}
The sets in $\mathcal{\ec} \cup \goodC$ are pairwise disjoint and
for every maximal end-component~$\ec$ of $G$ there exists a set $Y \supseteq \ec$
such that either $Y \in \mathcal{\ec}$ or $Y \in \goodC$.
\end{invariant}

\begin{algorithm2e}
	\SetAlgoRefName{MECImpr}
	\caption{Improved Algorithm for Maximal End-Components}
	\label{alg:mecimpr}
	\SetKwInOut{Input}{Input}
	\SetKwInOut{Output}{Output}
	\BlankLine
	\Input{an MDP~$\mdp = (G = (V, E), (\vp, \vr))$
	}
	\Output
	{
	 the set of maximal end-components of $\mdp$
	}
	\BlankLine
	$\mathcal{\ec} \gets \sccalg(G)$; $\goodC \gets \emptyset$\;
	\ForEach{$\scc \in \mathcal{\ec}$}{
		$T_\scc \gets \emptyset$\;
	}
	\While{$\mathcal{\ec} \ne \emptyset$}{
		remove some $S \in \mathcal{\ec}$ from $\mathcal{\ec}$\;
		$\rout \gets S \cap \vr \cap \pre(V\setminus S)$\label{l:mec:remb}\;
		$A \gets \at{R}{G}{\rout}$\label{l:mecimpr:remattr}\;
		$S \gets S \setminus A$\label{l:mecimpr:rem1}\;
		$T_S \gets (T_S \cup \pre(A)) \cap S$\label{l:mec:reme}\;
		\If(\tcc*[h]{$G[S]$ contains at least one edge}){$\post(S) \cap S \ne \emptyset$\label{l:mec:p1}}{ 
			\If{$\lvert T_S \rvert = 0$}{
				$\goodC \gets \goodC \cup \set{S}$\label{l:mec:good1}\;
			}\ElseIf{$\lvert T_S \rvert \ge \sqrt{m}$}{
				delete $T_S$\;
				$\mathcal{\scc} \gets \sccalg(G[S])$\;
				\If{$\lvert \mathcal{\scc} \rvert = 1$}{
					$\goodC \gets \goodC \cup \set{S}$\label{l:mec:good2}
				}\Else{
					\ForEach{$\scc \in \mathcal{\scc}$}{
								$T_\scc \gets \emptyset$\;}
					$\mathcal{\ec} \gets \mathcal{\ec} \cup \mathcal{\scc}$\;
				}
			}\Else{
				$\scc \gets $\ref{proc:lockstep}($G$, $S$, $\emptyset$, $T_S$)\;
				\If(\tcc*[h]{$G[\scc]$ contains at least one edge}){$\post(\scc) \cap \scc \ne \emptyset$\label{l:mec:p2}}{
					$\goodC \gets \goodC \cup \set{\scc}$\label{l:mec:good3}\;
				}
				$S \gets S \setminus \scc$\label{l:mecimpr:rem2}\;
				$T_S \gets (T_S \cup \pre(\scc)) \cap S$\label{l:mec:p3}\;
				$\mathcal{\ec} \gets \mathcal{\ec} \cup \set{S}$\;
			}
		}
	}
	\Return{$\goodC$}\;
\end{algorithm2e}

For each vertex set~$S$ in $\mathcal{\ec}$
additionally a subset $T_S$ of $S$ is maintained that contains vertices that
have lost outgoing edges since the last time a superset of $S$ was identified as 
strongly connected.
We use the following restrictions of Invariant~\ref{inv:HT} and Theorem~\ref{thm:lss}
(presented in Section~\ref{sec:lss}) to bottom SCCs only.

\begin{invariant}[Start Vertices BSCC]\label{inv:T}
	Either \upbr{a} $T_S$ is empty and 
	$G[S]$ is strongly connected or \upbr{b} at least one vertex of
	each bottom SCC of $G[S]$ is contained in $T_S$.
\end{invariant}

{
\begin{thm}[Lock-Step Search BSCC]\label{thm:lockstepbscc}
	Provided Invariant~\ref{inv:T} holds, 
	Procedure~\ref{proc:lockstep}\upbr{$G$, $S$, $\emptyset$, $T_S$}
	returns a bottom SCC~$\scc \subseteq S$ of
	$G[S]$ in $O(\lvert T_S \rvert \cdot \lvert \scc \rvert)$ symbolic steps.
\end{thm}

\begin{proof}
The proof of Theorem~\ref{thm:lockstepbscc} is a straightforward
simplification of the proof of Theorem~\ref{thm:lss} located in Appendix~\ref{sec:applss}.
\qed
\end{proof}
}

Initially the sets $T_S$ are empty.
The algorithm maintains Invariant~\ref{inv:T} for 
all $S \in \mathcal{\ec}$. 
This will ensure the correctness and the number of symbolic steps
of Procedure~\ref{proc:lockstep} (Section~\ref{sec:lss}) as called by the algorithm.

In each iteration of the while-loop one vertex set~$S$ is removed from $\mathcal{\ec}$ 
and processed. First the random vertices of $S$
with edges to vertices of $V \setminus S$ are identified and their random attractor
is removed from $S$.
After this step, there are no random vertices with edges 
from $S$ to $V \setminus S$. The predecessors of the removed vertices that are contained
in $S$ are added to $T_S$ and additionally $T_S$ is updated to only include 
vertices that are still in $S$. This preserves Invariant~\ref{inv:T} (see also 
\cite[Lemma 4.5.2]{Loitzenbauer16}).  The number of symbolic steps for the 
attractor computation can be charged to the removed vertices and is therefore bounded 
by $O(n)$ in total.

If afterwards $G[S]$ does not contain an edge anymore, 
then $S$ is not considered further and the algorithm continues with
the next iteration. Otherwise one of three cases happens.

\smallskip\noindent Case~(1): If $T_S$ is empty,
then by Invariant~\ref{inv:T} $G[S]$ is strongly connected, contains at least one 
edge and does not contain a random vertex with edges to $V \setminus S$, i.e.,
$S$ is an end-component, and by Invariant~\ref{inv:meccontained} it is a MEC.
In this case the algorithm adds the set $S$ to $\goodC$, which preserves 
both invariants and can happen at most $O(n)$ times.

\smallskip\noindent Case~(2): If there 
are at least $\sqrt{m}$ vertices in $T_S$, then the set $T_S$ is deleted and
as in the basic algorithm all SCCs of $G[S]$ are computed and add to $\mathcal{\ec}$
as new candidates for MECs. For each of the SCCs~$\scc$ a set $T_\scc$ is 
initialized with the empty set. As a vertex is added to a set $T_S$ only if one 
of its incoming edges is removed by the algorithm, Case~(2) can happen only $O(\sqrt{m})$
times over the whole algorithm. Thus the total number of symbolic steps 
for this case is $O(n \sqrt{m})$. Note that the Invariants~\ref{inv:T} 
and~\ref{inv:meccontained} are preserved.

\smallskip\noindent Case~(3): If $T_S$ contains less than $\sqrt{m}$ vertices, then 
Procedure~\ref{proc:lockstep}($G$, $S$, $\emptyset$, $T_S$) is called. 
By Invariant~\ref{inv:T} and Theorem~\ref{thm:lockstepbscc} the procedure
returns a bottom SCC~$\scc$ of $G[S]$ in $O(\lvert T_S \rvert \cdot \lvert \scc \rvert)$
many symbolic steps. Since there are no random edges between $S$ and $V\setminus S$
in $\mdp$ and $\scc$ has no outgoing edges in $G[S]$, we have that $\scc$ is 
an end-component if it contains at least one edge. By Invariant~\ref{inv:meccontained}
it is also a MEC and is correctly added to $\goodC$. 
As the sets in $\goodC$ are not considered further by the algorithm, we can 
charge the symbolic steps of Procedure~\ref{proc:lockstep} to the vertices
of $\scc$. Thus this part takes at most $O(n \sqrt{m})$ symbolic steps over the 
whole algorithm. The vertices of 
$S \setminus \scc$ are added back to $\mathcal{\ec}$, which preserves 
Invariant~\ref{inv:meccontained}. The predecessors of $\scc$ in
$S \setminus \scc$  are added to $T_{S \setminus \scc}$ and vertices of $\scc$ 
are removed from $T_{S \setminus \scc}$, which preserves Invariant~\ref{inv:T}.

By the above case analysis we have that each vertex set that is added to 
$\goodC$ is indeed a MEC (soundness). By Invariant~\ref{inv:meccontained}
and $\mathcal{\ec} = \emptyset$ at termination of the algorithm we
further have completeness.
In each iteration either $S$ does not contain an edge and is not considered further,
a set is added to $\goodC$ (and not contained in $\mathcal{\ec}$ after that) or 
case~(2) happens. Thus there are at most $O(n + \sqrt{m})$ iterations
of the algorithm. The symbolic operations we have not yet accounted for in 
the analysis of the number of symbolic steps are of $O(1)$ per iteration. Hence
Algorithm~\ref{alg:mecimpr} takes $O(n \sqrt{m})$ symbolic steps and correctly
computes the MECs of the given MDP~$\mdp$.

\begin{lemma}[Invariants of Improved Algorithm for MEC]\label{lem:improvedmecinv}
Invariant~\ref{inv:T} and Invariant~\ref{inv:meccontained} are preserved throughout
Algorithm~\ref{alg:mecimpr}, i.e., they hold before the
first iteration, after each iteration, and after termination of the main while-loop. Further,
Invariant~\ref{inv:T} is preserved during each iteration of the main while-loop.
\end{lemma}

\begin{proof}
\item 
\smallskip\noindent{\bf Invariant~\ref{inv:T}.}
The proof of maintaining Invariant~\ref{inv:T} in Algorithm~\ref{alg:mecimpr} is a
straightforward simplification of the proof of maintaining Invariant~\ref{inv:HT} in
Algorithm~\ref{alg:streettgraphimpr} (located in Appendix~\ref{sec:appgraphs}).

\smallskip\noindent{\bf Invariant~\ref{inv:meccontained} -- Disjointness.}
The sets in $\mathcal{\ec} \cup \goodC$ are pairwise disjoint at the
initialization since $\goodC$ is initialized as $\emptyset$. Furthermore,
whenever a set $S$ is added to $\goodC$ in an iteration of the main
while-loop, a superset $\tilde{S} \supseteq S$ is removed from
$\mathcal{\ec}$ in the same iteration of the while-loop. Therefore
by induction the disjointness of the sets in $\mathcal{\ec} \cup \goodC$
is preserved.

\smallskip\noindent{\bf Invariant~\ref{inv:meccontained} -- Containment of maximal end-components.}
At initialization, $\mathcal{\ec}$ contains all SCCs of $G$.
Each maximal end-component $X$ of $\mdp = (G = (V, E), \allowbreak(\vp, \vr), \allowbreak\trans)$
is strongly connected, so there exists an SCC $Y \supseteq X$ of $G$
such that $Y \in \mathcal{\ec}$.

Consider a set $S \in \mathcal{\ec}$ that is removed from $\mathcal{\ec}$
at the beginning of an iteration of the main while-loop. Consider further
a maximal end-component $X$ of $\mdp$ such that $X \subseteq S$. We require
that a set $Y \supseteq X$ is added to either
$\mathcal{\ec}$ or $\goodC$ in this iteration of the main while-loop.

First, after we remove $\at{R}{G}{S \cap \vr \cap \pre(V\setminus S)}$ from $S$, we maintain the fact
that $X \subseteq S$ by Lemma~\ref{lem:eccontained}. Second, $G[S]$ contains an edge
since $X \subseteq S$. Finally, one of the three cases happens:

\smallskip\noindent Case~(1): If $\lvert T_S \rvert = 0$, then the set
$S \supseteq X$ is added to $\goodC$.

\smallskip\noindent Case~(2): If $\lvert T_S \rvert \ge \sqrt{m}$,
then the algorithm computes the SCCs of $G[S]$. Since $X \subseteq S$ is strongly connected,
it is completely contained in some SCC $Y$ of $G[S]$, and $Y$ is added to~$\mathcal{\ec}$.

\smallskip\noindent Case~(3): If $0 < \lvert T_S \rvert < \sqrt{m}$,
then the algorithm partitions $S$ into $C$ and $S \setminus C$. By Theorem~\ref{thm:lockstepbscc} we
have that $C$ is a (bottom) SCC of $G[S]$. Since $X \subseteq S$ is strongly connected, it is completely
contained in some SCC of $G[S]$. Therefore either $X \subseteq C$ or $X \subseteq (S \setminus C)$.
The set $S \setminus C$ is added to $\mathcal{\ec}$. If $X \subseteq C$, then in particular
$G[C]$ contains an edge, and $C$ is added to $\goodC$.

\smallskip\noindent By the above case analysis we have that a set $Y \supseteq X$ is added
to either $\mathcal{\ec}$ or $\goodC$ in the iteration of the main while-loop.
\qed
\end{proof}

\begin{proof}[of Theorem~\ref{thm:improvedmecs}]

\smallskip\noindent{\bf Correctness.}
A candidate set can be added to $\goodC$ in three cases. When $S$ is added to $\goodC$
at line~\ref{l:mec:good1} (resp. at line~\ref{l:mec:good2}), then it contains an edge by the check at line~\ref{l:mec:p1},
it is strongly connected by $\lvert T_S \rvert = 0$ and Invariant~\ref{inv:T} (resp. by the result of $\sccalg$), and it has no
random vertices with edges to $V \setminus S$ by the random attractor removal at lines~\ref{l:mec:remb}-\ref{l:mec:reme}.
When $C$ is added at line~\ref{l:mec:good3}, then it contains an edge by the check at line~\ref{l:mec:p2},
it is strongly connected by Theorem~\ref{thm:lockstepbscc}, it contains no random vertices with edges to
$V \setminus S$ by the random attractor removal at lines~\ref{l:mec:remb}-\ref{l:mec:reme}, and it
contains no random vertices with edges to $S \setminus C$ by the fact that $C$ is a bottom SCC of $G[S]$
(see Theorem~\ref{thm:lockstepbscc}). Therefore we have that whenever a candidate set is added to $\goodC$,
it is an end-component, and by induction and Invariant~\ref{inv:meccontained} we have that
it is a maximal end-component (soundness).

Finally, by soundness, Invariant~\ref{inv:meccontained}, the termination of the algorithm
(shown below), and the fact that $\mathcal{\ec} = \emptyset$ at the
termination of the algorithm, we have that $\goodC$ contains all the maximal
end-components of $\mdp$ (completeness).

\smallskip\noindent{\bf Symbolic steps analysis.}
By~\cite{GentiliniPP08}, the initialization with the SCCs of
a given MDP takes $O(n)$ symbolic steps.

In each iteration of the outer while-loop, a set~$S$ is removed 
from $\mathcal{\ec}$ and \upbr{a} $S$ is added to $\goodC$, or
\upbr{b} at least two sets that are (subsets of) a partition of $S$ are added
to $\mathcal{\ec}$, or \upbr{c} $S$ is partitioned into two sets, one
of them may be added to $\goodC$ and the other is added to $\mathcal{\ec}$.
All three cases can happen at most $O(n)$ times, so there can be at most
$O(n)$ iterations of the outer while-loop. The $\pre$ and $\post$ operations
at lines~\ref{l:mec:remb},~\ref{l:mec:reme},~\ref{l:mec:p1},~\ref{l:mec:p2},
and~\ref{l:mec:p3} can be charged to the iterations of the outer while-loop.

Each $\cpre{R}$ operation executed as a part of the random attractor computation
at line~\ref{l:mecimpr:remattr} adds at least one vertex to $A$, and the vertices
of $A$ are then not considered any further in the algorithm. Therefore there
can, in total, be at most $O(n)$ $\cpre{R}$ operations over all attractor computations
at line~\ref{l:mecimpr:remattr}.

Note that every vertex in each of $T_S$ can be attributed to at least one unique implicit edge deletion 
since we only add vertices to $T_S$ that are predecessors of the vertices that were separated from $S$ 
(or deleted from the maintained graph). Whenever the case $\lvert T_S \rvert \ge \sqrt{m}$
occurs, for all subsets $C \subseteq S$ that are then added to $\mathcal{\ec}$, we initialize $T_C = \emptyset$.
Therefore, the case $\lvert T_S \rvert \ge \sqrt{m}$ can happen at most $O(\sqrt{m})$ times throughout
the algorithm since there are at most $m$ edges that can be deleted. By~\cite{GentiliniPP08} we have a bound
$O(n)$ for one iteration, so we can bound the total number of symbolic steps in all iterations of this
case by $O(n \cdot \sqrt{m})$.

It remains to bound the number of symbolic steps in Procedure~\ref{proc:lockstep}. Let $\scc$ be the set
returned by $\ref{proc:lockstep}$($G$, $S$, $\emptyset$, $T_S$). By Theorem~\ref{thm:lockstepbscc} and the fact that
$\lvert T_S \rvert < \sqrt{m}$, the number of symbolic steps in this call is bounded by $O(\sqrt{m} \cdot \lvert \scc \rvert)$,
and the set $\scc$ is not considered further in the algorithm after this call. Hence we can bound the total number of
symbolic steps in all calls of the procedure by $O(n \cdot \sqrt{m})$.
\qed
\end{proof}

\section{Details of Section~\ref{sec:mdps}: MDPs with Streett Objectives}\label{sec:appmdps}

\subsection{Basic Symbolic Algorithm for MDPs with Streett Objectives}

\begin{algorithm2e}
	\SetAlgoRefName{StreettMDPbasic}
	\caption{Basic Algorithm for MDPs with Streett Obj.}
	\label{alg:streettmdpbasic}
	\SetKwInOut{Input}{Input}
	\SetKwInOut{Output}{Output}
	\BlankLine
	\Input{MDP $\mdp = ((V, E), (\vp, \vr), \trans)$ and pairs $\SP= \{(L_i, U_i) \mid 1 \le i \le k\}$
	}
	\Output
	{
	$\as{\mdp, \Streett{\SP}}$
	}
	\BlankLine
	$\mathcal{\ec} \gets \mecalg(\mdp)$; $\good \gets \emptyset$\;
	\While{$\mathcal{\ec} \ne \emptyset$}{
		remove some $S \in \mathcal{\ec}$ from $\mathcal{\ec}$\;
		$\badv \gets \bigcup_{1 \le i \le k : U_i \cap S = \emptyset} (L_i \cap S)$\;
		\If{$\badv \ne \emptyset$}{
			$S \gets S \setminus \at{R}{\mdp[S]}{\badv}$\;
			$\mathcal{\ec} \gets \mathcal{\ec} \cup \mecalg(\mdp[S])$\;
		}\lElse{
			$\good \gets \good \cup \set{S}$
		}
	}
	\Return{$\as{\mdp, \Reach{\bigcup_{\ec \in \good} \ec}}$}\;
\end{algorithm2e}

The pseudocode of the basic symbolic algorithm for MDPs with Streett objectives
is given in Algorithm~\ref{alg:streettmdpbasic}. The key differences
compared to Algorithm~\ref{alg:streettgraphbasic} are as follows: 
(a)~SCC computation is replaced by MEC computation;
(b)~along with the removal of bad vertices, their random attractor is also removed;
and (c)~removing the attractor ensures that the check required for trivial SCCs
for graphs (line~\ref{l:gbasic:edge})
is not required any further.

To compute the almost-sure winning set for MDPs with Streett objectives,
we first find all (maximal) good end-components and then solve almost-sure
reachability with the union of the good end-components as target set as the last 
step of the algorithm. This is correct by Lemma~\ref{lem:gec}. Towards finding 
all good end-components, the algorithm maintains two sets, the set $\good$ of
identified good end-components that is initially empty and the set $\mathcal{\ec}$ 
of end-components that are candidates for good end-components that is 
initialized with the MECs of the MDP.
In each iteration of the while-loop one set $S$ is removed from the set of 
candidates $\mathcal{\ec}$ and the set of bad vertices $\bad(S)$ of $S$
is determined. If $\bad(S)$ is empty, then $S$ is a good end-component and 
added to $\good$. Otherwise the random attractor of $\bad(S)$ in $\mdp[S]$
is removed from $S$, which by Corollary~\ref{cor:geccontained} does not 
remove any vertices that are in a good end-component. The remaining vertices of 
$S$ have no outgoing random edges and thus still induce a sub-MDP but the 
sub-MDP might not be strongly connected any more. Then the MECs
of this sub-MDP are added to $\mathcal{\ec}$.
These operations maintain the invariants that (i) each set in $\mathcal{\ec}$
is an end-component and (ii) each good end-component is a subset of one set 
in either $\good$ or $\mathcal{\ec}$. By (i) a set in $\mathcal{\ec}$ is a (maximal) 
good end-component if it does not contain any bad vertices, i.e., in particular
this holds for the sets added to $\good$ (soundness). By (ii) and 
$\mathcal{\ec} = \emptyset$ at termination of the while-loop the algorithm
identifies all (maximal) good end-components of the MDP (completeness).
Since in each iteration of the while-loop either (1) a set is removed from $\mathcal{\ec}$
and added to $\good$ or (2) bad vertices are removed from a set and not considered 
further by the algorithm, there can be at most $O(n)$ iterations of the while-loop. 
Furthermore, whenever bad vertices are removed, then the number of target pairs
a given candidate set intersects is reduced by one. Thus each vertex 
is considered in at most $O(k)$ iterations of the while-loop.
The most expensive operation in the while-loop is the computation of the MECs. 
Denoting the number of symbolic steps for the MEC computation with $O(\mectime)$,
the number of symbolic steps of Algorithm~\ref{alg:streettmdpbasic} 
is $O(\min(n, k) \cdot \mectime)$ (assuming that the number of symbolic steps 
for the almost-sure reachability computation is lower than that).

\subsection{Improved Symbolic Algorithm for MDPs with Streett Objectives}

We present the technical details regarding the improved symbolic algorithm
for MDPs with Streett objectives. The main ideas of the algorithm are presented
in Section~\ref{sec:mdps}. The pseudocode is given in Algorithm~\ref{alg:streettmdpimpr}.


\begin{algorithm2e}
	\SetAlgoRefName{StreettMDPimpr}
	\caption{Improved Alg. for MDPs with Streett Obj.}
	\label{alg:streettmdpimpr}
	\SetKwInOut{Input}{Input}
	\SetKwInOut{Output}{Output}
	\BlankLine
    \Input{MDP $\mdp = ((V, E), (\vp, \vr), \trans)$ and pairs $\SP= \{(L_i, U_i) \mid 1 \le i \le k\}$
    }
	\Output
	{
	$\as{\mdp, \Streett{\SP}}$
	}
	\BlankLine
	$\mathcal{\ec} \gets \mecalg(\mdp)$; $\good \gets \emptyset$\;
	\lForEach{$\scc \in \mathcal{\ec}$}{
		$H_\scc \gets \emptyset$; $T_\scc \gets \emptyset$
	}
	\While{$\mathcal{\ec} \ne \emptyset$}{
		remove some $S \in \mathcal{\ec}$ from $\mathcal{\ec}$\;
		$\badv \gets \bigcup_{1 \le i \le k : U_i \cap S = \emptyset} (L_i \cap S)$\;
		\While{$\badv \ne \emptyset$\label{l:imdp:innerw}}{
			$A \gets \at{R}{\mdp[S]}{\badv}$\label{l:imdp:attr1}\;
			$S \gets S \setminus A$\label{l:mimpr:rem1}\;
			$H_S \gets (H_S \cup \post(A)) \cap S$\label{l:imdp:bh}\;
			$T_S \gets (T_S \cup \pre(A)) \cap S$\label{l:imdp:bt}\;
			$\badv \gets \bigcup_{1 \le i \le k : U_i \cap S = \emptyset} (L_i \cap S)$\;
		}
		\If(\tcc*[h]{$\mdp[S]$ contains at least one edge}){$\post(S) \cap S \ne \emptyset$\label{l:imdp:edge}}{
			\lIf{$\lvert H_S \rvert + \lvert T_S \rvert = 0$}{
				$\good \gets \good \cup \set{S}$\label{l:imdp:good1}
			}\ElseIf{$\lvert H_S \rvert + \lvert T_S \rvert \ge \sqrt{m / \log n}$}{
				delete $H_S$ and $T_S$\;
				$\mathcal{\scc} \gets \sccalg(\mdp[S])$\;
				\lIf{$\lvert \mathcal{\scc} \rvert = 1$}{
					$\good \gets \good \cup \set{S}$\label{l:imdp:good2}
				}\Else{
					\ForEach{$\scc \in \mathcal{\scc}$}{
						$\rout \gets \scc \cap \vr \cap \pre(S \setminus \scc)$\label{l:imdp:sccrout}\;
						$A \gets \at{R}{\mdp[\scc]}{\rout}$\label{l:imdp:attr2}\;
						$\scc \gets \scc \setminus A$\label{l:mimpr:rem3}\;
						$H_\scc \gets \post(A) \cap \scc$\label{l:imdp:scch}\;
						$T_\scc \gets \pre(A) \cap \scc$\label{l:imdp:scct}\;
						$\mathcal{\ec} \gets \mathcal{\ec} \cup \set{\scc}$\;
					}
				}
			}\Else{
				($\scc$, $H_S$, $T_S$) $\gets $\ref{proc:lockstep}($G$, $S$, $H_S$, $T_S$)\;
				\lIf{$\scc = S$}{
					$\good \gets \good \cup \set{S}$\label{l:imdp:good3}
				}\Else(\tcc*[h]{separate $\scc$ and $S \setminus \scc$}){
					$\rout_\scc \gets \scc \cap \vr \cap \pre(S \setminus \scc)$\label{l:mimpr:lssrout}\tcc*[h]{empty if $\scc$ bottom SCC}
					$A_\scc \gets \at{R}{\mdp[\scc]}{\rout_\scc}$\label{l:imdp:attr3}\tcc*[h]{$=\at{R}{\mdp[S]}{S \setminus \scc} \cap \scc$}
					$A_S \gets \at{R}{\mdp[S]}{\scc}$\label{l:imdp:attr4}\;
					$\scc \gets \scc \setminus A_\scc$\label{l:mimpr:rem4}\;
					$S \gets S \setminus A_S$\label{l:mimpr:rem2}\;
					$H_\scc \gets \post(A_\scc) \cap \scc$\label{l:imdp:lsshc}\;
					$T_\scc \gets \pre(A_\scc) \cap \scc$\label{l:imdp:lsstc}\;
					$H_S \gets (H_S \cup \post(A_S)) \cap S$\label{l:imdp:lsshs}\;
					$T_S \gets (T_S \cup \pre(A_S)) \cap S$\label{l:imdp:lssts}\;
					$\mathcal{\ec} \gets \mathcal{\ec} \cup \set{S} \cup 
					\set{\scc}$\;
				}
			}
		}
	}
	\Return{$\as{\mdp, \Reach{\bigcup_{\scc \in \good} \scc}}$}\;
\end{algorithm2e}

The following invariant is maintained throughout Algorithm~\ref{alg:streettmdpimpr}
for the sets in $\good$ and $\mathcal{\ec}$.

\begin{invariant}[Maintained Sets]\label{inv:gecontained}
The sets in $\mathcal{\ec} \cup \good$ are pairwise disjoint and for every good 
end-component~$\scc$ of $G$ there exists a set $Y \supseteq \scc$
such that either $Y \in \mathcal{\ec}$ or $Y \in \good$.
\end{invariant}

Furthermore, the algorithm maintains the invariant
that each candidate for a good end-component $S \in \mathcal{\ec}$
contains no random edges to vertices not in $S$.

\begin{invariant}[No Random Outgoing Edges]\label{inv:mdprout}
Given an MDP~$\mdp$ and its underlying graph $G = (V, E)$, for each
set $S \in \mathcal{\ec}$ there are no random vertices in $S$ with edges
to vertices in~$V \setminus S$.
\end{invariant}

Finally, for each candidate set $S \in \mathcal{\ec}$ the algorithm remembers
sets $H_S$ and $T_S$ of vertices that have lost incoming resp. outgoing
edges since the last time a superset of $S$ was identified as being strongly connected.
The algorithm maintains Invariant~\ref{inv:HT} and therefore it can use
Procedure~\ref{proc:lockstep} together with its correctness guarantee and bound
on symbolic steps provided by Theorem~\ref{thm:lss}.

\begin{lemma}[Invariants of Improved Algorithm for MDPs]\label{lem:improvedmdpsinv}
Invariant~\ref{inv:HT}, Invariant~\ref{inv:gecontained}, and Invariant~\ref{inv:mdprout}
are preserved throughout Algorithm~\ref{alg:streettmdpimpr}, i.e., they hold before the
first iteration, after each iteration, and after termination of the main while-loop. Further,
Invariant~\ref{inv:HT} is preserved during each iteration of the main while-loop.
\end{lemma}

\begin{proof}
\item 
\smallskip\noindent{\bf Invariant~\ref{inv:HT}.}
The proof is a minor extension of the maintenance proof for Algorithm~\ref{alg:streettgraphimpr}
that is given in Appendix~\ref{sec:appgraphs}. In terms of strong connectivity of a candidate $S$
and the maintenance of the sets $H_S$ and $T_S$, the only difference to the graph case is that
after an SCC~$C$ is computed by $\sccalg$ or Procedure~\ref{proc:lockstep}, another subset
of vertices~$A$ (vertices with outgoing random edges and their random attractor) is removed
from $C$. In this case the invariant is maintained by initializing $H_C$ resp.\ $T_C$ with 
the vertices of $C \setminus A$ with edges from resp.\ to vertices of $A$, i.e., 
$H_C \gets \post(A) \cap C$ and $T_C \gets \pre(A) \cap C$.

\smallskip\noindent{\bf Invariant~\ref{inv:gecontained} -- Disjointness.}
The sets in $\mathcal{\ec} \cup \good$ are pairwise disjoint at the
initialization since $\good$ is initialized as $\emptyset$. Furthermore,
whenever a set $S$ is added to $\good$ in an iteration of the main
while-loop, a superset $\tilde{S} \supseteq S$ is removed from
$\mathcal{\ec}$ in the same iteration of the while-loop. Therefore
by induction the disjointness of the sets in $\mathcal{\ec} \cup \good$
is preserved.

\smallskip\noindent{\bf Invariant~\ref{inv:gecontained} -- Containment of good end-components.}
At initialization, $\mathcal{\ec}$ contains all MECs of the input MDP
$\mdp = (G = (V, E), \allowbreak(\vp, \vr), \allowbreak\trans)$.
Each good end-component $C$ of $P$ is an end-component, 
so there exists a MEC $Y \supseteq C$ such that $Y \in \mathcal{\ec}$
for each good end-component $C$.

Consider a set $S \in \mathcal{\ec}$ that is removed from $\mathcal{\ec}$
at the beginning of an iteration of the main while-loop. Consider further
a good end-component $C$ of $P$ such that $C \subseteq S$. We require that a set
$Y \supseteq C$ is added to either
$\mathcal{\ec}$ or $\good$ in this iteration of the main while-loop.

First, whenever we remove $\at{R}{\mdp[S]}{\bad(S)}$ from $S$, by Corollary~\ref{cor:geccontained},
we maintain the fact that $C \subseteq S$. Second, $P[S]$ contains an edge since $C \subseteq S$.
Finally, one of the three cases happens:

\smallskip\noindent Case~(1): If $\lvert H_S \rvert + \lvert T_S \rvert = 0$, then the set
$S \supseteq C$ is added to $\good$.

\smallskip\noindent Case~(2): If $\lvert H_S \rvert + \lvert T_S \rvert \ge \sqrt{m / \log n}$,
then the algorithm computes the SCCs of $\mdp[S]$. If $S$ itself is the (sole) SCC of $\mdp[S]$, then it
is added to $\good$. Otherwise, since $C \subseteq S$ is strongly connected,
it is completely contained in some SCC $Y$ of $\mdp[S]$. Furthermore, since $C$ has no outgoing
random edges, by Lemma~\ref{lem:eccontained} it is contained in $Y$ even after we remove
$\at{R}{\mdp[Y]}{ Y \cap \vr \cap \pre(S \setminus Y) }$ from it. Finally, $Y$ is added to $\mathcal{\ec}$.

\smallskip\noindent Case~(3): If $0 < \lvert H_S \rvert + \lvert T_S \rvert < \sqrt{m / \log n}$,
then the algorithm either adds $S \supseteq C$ to $\good$, or partitions $S$
into $\tilde{S}$ and $S \setminus \tilde{S}$. Suppose the latter case happens, then by
Theorem~\ref{thm:lss} we have that $\tilde{S}$ is an SCC of $\mdp[S]$. Further, since $C \subseteq S$
is strongly connected, it is completely contained in some SCC of $\mdp[S]$.
Therefore either $C \subseteq \tilde{S}$ or $C \subseteq (S \setminus \tilde{S})$.
If $C \subseteq \tilde{S}$, then by Lemma~\ref{lem:eccontained} after the removal
of $\at{R}{\mdp[\tilde{S}]}{ \tilde{S} \cap \vr \cap \pre(S \setminus \tilde{S}) }$ from
$\tilde{S}$ we maintain that $C \subseteq \tilde{S}$.
If $C \subseteq (S \setminus \tilde{S})$, then by Lemma~\ref{lem:eccontained} after the removal
of $\at{R}{\mdp[S]}{ \tilde{S} }$ from $(S \setminus \tilde{S})$ we maintain that $C \subseteq (S \setminus \tilde{S})$.
Finally, both $\tilde{S}$ and $S \setminus \tilde{S}$ are added to $\mathcal{\ec}$.

\smallskip\noindent By the above case analysis we have that a set $Y \supseteq C$ is added
to either $\mathcal{\ec}$ or $\good$ in the iteration of the main while-loop.

\smallskip\noindent{\bf Invariant~\ref{inv:mdprout}.}
Given an MDP, the set $\mathcal{\ec}$ is initialized with the MECs of the MDP, and
by definition they have no random outgoing edges. Therefore the invariant holds
before the first iteration of the main while-loop.

Consider a candidate set $S \in \mathcal{\ec}$ in a given iteration of the main
while-loop. By the induction hypothesis, $S$ has no random vertices with edges to $V \setminus S$.
First, some bad vertices can be iteratively removed from $S$. At each such removal,
the random attractor to these vertices is removed from $S$ as well. After the removal,
by the definition of a random attractor, $S$ has no random outgoing edges to the attractor,
and therefore by induction has no random outgoing edges to $V \setminus S$.
Second, $S$ may be partitioned into at least two proper subsets. Then for each
such subset $C$, the random attractor to random vertices in $C$ with edges
to $S \setminus C$ is removed from $C$. By induction and the definition of a random attractor,
after the removal $C$ contains no random outgoing edges to $V \setminus C$ and 
adding it to $\mathcal{\ec}$ preserves the invariant.

\qed
\end{proof}

\begin{proof}[of Theorem~\ref{thm:improvedmdps}]

\smallskip\noindent{\bf Correctness.}
Whenever a candidate set $S$ is added to $\good$, it contains an edge
by the check at line~\ref{l:imdp:edge}, $\bad(S) = \emptyset$ by the check
at line~\ref{l:imdp:innerw}, and it has no outgoing random edges
by Invariant~\ref{inv:mdprout} and the random attractor
removal at line~\ref{l:mimpr:rem1}.
Furthermore, \upbr{a} at line~\ref{l:imdp:good1}, $S$ is strongly connected by Invariant~\ref{inv:HT},
\upbr{b} at line~\ref{l:imdp:good2}, $S$ is strongly connected by the result of~\sccalg, and \upbr{c}
at line~\ref{l:imdp:good3}, $S$ is strongly connected by Theorem~\ref{thm:lss}.
Therefore we have that whenever a candidate set is added to $\good$,
it is indeed a good end-component (soundness).

Finally, by soundness, Invariant~\ref{inv:gecontained}, the termination of the
algorithm (shown below), and the fact that $\mathcal{\ec} = \emptyset$ at the
termination of the algorithm, we have that $\good$ contains all good
end-components of $G$ (completeness).

\smallskip\noindent{\bf Symbolic steps analysis.}
When using our improved symbolic algorithm for MEC decomposition,
the initialization takes $O(n\cdot \sqrt{m})$ symbolic steps by Theorem~\ref{thm:improvedmecs}.

In each iteration of the outer while-loop, a set~$S$ is removed
from $\mathcal{\ec}$ and either \upbr{a} a set $S' \subseteq S$ is added
to $\good$ and no set is added to $\mathcal{\ec}$ or \upbr{b} at least two
sets that are (subsets of) a partition of $S$ are added to $\mathcal{\ec}$.
Both can happen at most $O(n)$ times, thus there can be at most $O(n)$ iterations
of the outer while-loop. The $\pre$ and $\post$ operations at
lines~\ref{l:imdp:edge},~\ref{l:mimpr:lssrout},~\ref{l:imdp:lsshc},~\ref{l:imdp:lsstc},~\ref{l:imdp:lsshs},
and~\ref{l:imdp:lssts} can be charged to the iterations of the outer while-loop.

An iteration of the inner while-loop (line~\ref{l:imdp:innerw}) is executed
only if some vertices~$\badv$ are removed from $S$; the vertices of $\badv$
are then not considered further. Thus there can, in total, be at most $O(n)$
$\post$ operations at line~\ref{l:imdp:bh} and $\pre$ operations at line~\ref{l:imdp:bt}
over all iterations of the inner while-loop.

Similarly, each $\cpre{R}$ operation executed as a part of a random attractor computation
adds at least one vertex to the attractor, and the vertices of the attractor are then not
considered any further in the algorithm. Therefore there can, in total, be at most
$O(n)$ $\cpre{R}$ operations over all attractor computations at
lines~\ref{l:imdp:attr1},~\ref{l:imdp:attr2},~\ref{l:imdp:attr3}, and~\ref{l:imdp:attr4}.

Note that every vertex in each of $H_S$ and $T_S$ can be attributed to at least one unique
implicit edge deletion since we only add vertices to $H_S$ resp.\ $T_S$ that are successors
resp.\ predecessors of vertices that were separated from $S$ (or deleted from the
maintained graph). Whenever the case $\lvert H_S \rvert + \lvert T_S \rvert \ge \sqrt{m / \log{n}}$
occurs, for all subsets $C \subseteq S$ that are then added to $\mathcal{\ec}$, we initialize
$H_C = T_C = \emptyset$. Therefore, the case $\lvert H_S \rvert + \lvert T_S \rvert \ge \sqrt{m / \log{n}}$
can happen at most $O(\sqrt{m \log{n}})$ times throughout the algorithm since
there are at most $m$ edges that can be deleted. In one iteration of this case, the number
of symbolic steps executed by $\sccalg$ together with symbolic steps executed at
lines~\ref{l:imdp:sccrout},~\ref{l:imdp:scch}, and~\ref{l:imdp:scct}, is bounded by $O(n)$~\cite{GentiliniPP08}.

It remains to bound the number of symbolic steps in Procedure~\ref{proc:lockstep}. Let $\scc$ be the set
returned by the procedure; we charge the symbolic steps in this call of the procedure to the vertices of 
the smaller set of $\scc$ and $S \setminus \scc$. By Theorem~\ref{thm:lss} we have either \upbr{a}
$\scc = S$, the number of symbolic steps in this call is bounded by $O(\sqrt{m/\log{n}} \cdot \lvert \scc \rvert)$,
and the set $S$ is added to $\good$ or \upbr{b}
$\min(\lvert \scc \rvert, \lvert S \setminus \scc \rvert) \le \lvert S \rvert / 2$ and the number of symbolic steps
in this call is bounded by $O(\sqrt{m/\log{n}} \cdot \min(\lvert \scc \rvert, \lvert S \setminus \scc \rvert))$.
Case~\upbr{a} can happen at most once for the vertices of $\scc$, and for case~\upbr{b} note that the size
of a set containing a specific vertex can be halved at most $O(\log{n})$ times; thus we charge each vertex at 
most $O(\log{n})$ times. Hence we can bound the total number of symbolic steps in all calls to the procedure
by $O(n \cdot \sqrt{m \log{n}})$.
\qed
\end{proof}

\section{Details of Section~\ref{sec:exper}: Experiments}\label{sec:appexper}

We present the results of the experimental evaluation when comparing based on the time.
In all the figures, both axes plot the amount of seconds spent on the execution. Similar
to the case of symbolic steps, we begin the measurement after the initial preprocessing
step (computing all SCCs for graphs and all MECs for MDPs) is finished. The results for
graphs are shown in Figure~\ref{fig:graphsT} and the results for MDPs are shown in 
Figure~\ref{fig:mdpsT}.

\newpage

\setlength{\abovecaptionskip}{2pt}

\begin{figure}[t]
\centering
\includegraphics[width=0.7\textwidth]{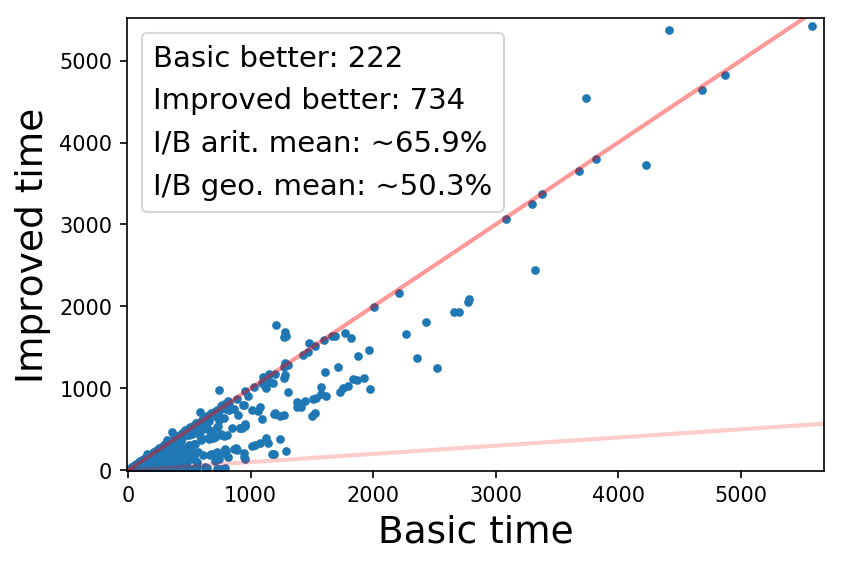}
\caption{Comparison of time for graphs with Streett objectives.}
\label{fig:graphsT}
\end{figure}

\setlength{\abovecaptionskip}{4pt}

\begin{figure}
\begin{center}
     \subfloat[10\% random vertices\label{fig:mdp10T}]{%
       \includegraphics[width=0.5\textwidth]{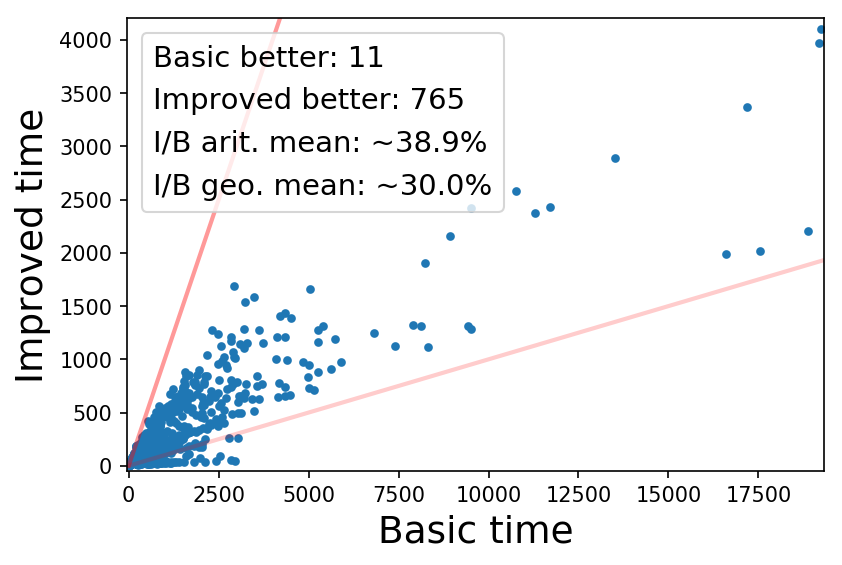}
     }
     \subfloat[20\% random vertices\label{fig:mdp20T}]{%
       \includegraphics[width=0.5\textwidth]{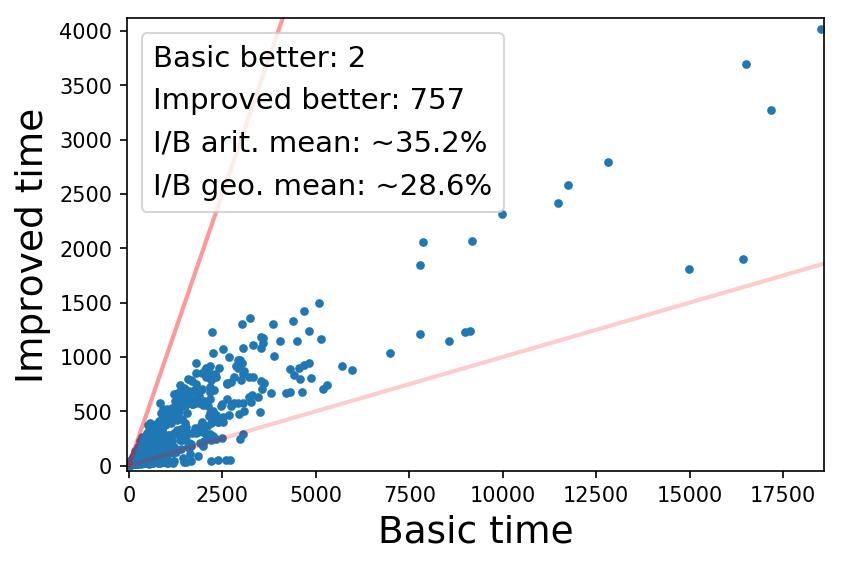}
     }
     \hfill
     \subfloat[50\% random vertices\label{fig:mdp50T}]{%
       \includegraphics[width=0.5\textwidth]{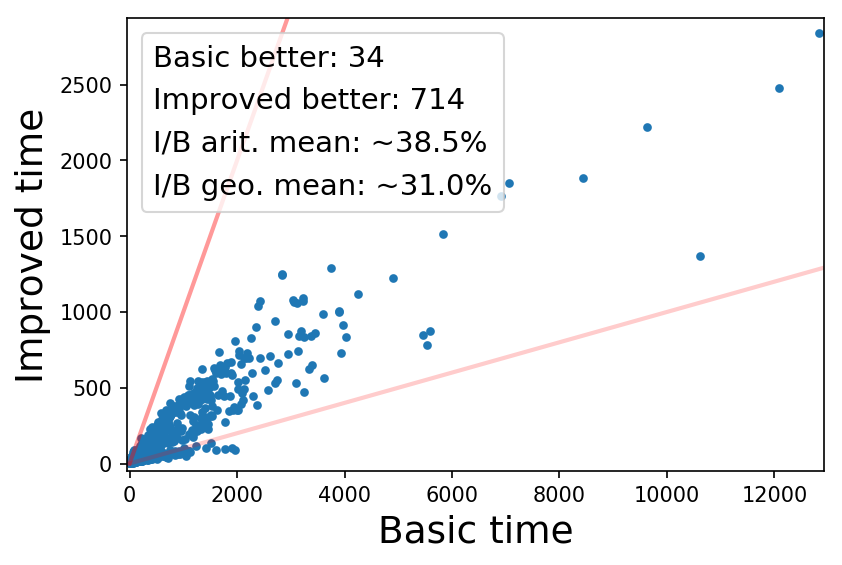}
     }
\caption{Comparison of time for MDPs with Streett objectives.}
\label{fig:mdpsT}
\end{center}
\end{figure}

\end{document}